\RequirePackage{fix-cm}
\documentclass[journal,10pt,onecolumn,draftclsnofoot]{IEEEtran}

\usepackage{color}
\usepackage{epsfig}
\usepackage{amsmath}
\usepackage{amsfonts}
\usepackage{amsthm}
\usepackage{amssymb}
\usepackage{mathrsfs}
\usepackage{amssymb}
\usepackage{graphics}
\usepackage{graphicx}
\usepackage{float}
\usepackage{algorithm}
\usepackage{algorithmic}
\usepackage{subcaption}

\newtheorem{example}{Example}
\newtheorem{definition}{Definition}
\newtheorem{theorem}{Theorem}
\newtheorem{lemma}{Lemma}

\newtheorem{corollary}{Corollary}

\newcommand{\Z}{\mathbb{Z}}
\newcommand{\R}{\mathbb{R}}

\newcommand{\ning}{\Gamma(V,\alpha,\beta,k) }
\newcommand{\turing}{\Gamma(2S,T_{2\theta},T,3) }

\DeclareMathOperator*{\vol}{vol}

\title{On the decoding of lattices constructed via a single parity check}

\author{Vincent Corlay, Joseph J. Boutros, Philippe Ciblat, and Lo\"ic Brunel 
\thanks{V. Corlay is with Mitsubishi Electric R\&D Centre Europe, Rennes, France, and Telecom Paris,  Palaiseau, France (v.corlay@fr.merce.mee.com). 
J. J. Boutros is with the Department of Electrical and Computer Engineering, Texas A\&M University at Qatar, Doha, Qatar (boutros@tamu.edu).
P.~Ciblat is with Telecom Paris, Palaiseau, France (philippe.ciblat@telecom-paris.fr). 
L. Brunel is with Mitsubishi Electric R\&D Centre Europe, Rennes, France (l.brunel@fr.merce.mee.com).
}
\thanks{Part of the section on Barnes-Wall lattices was presented at the IEEE International Symposium on Information Theory, Los Angeles, USA, July 2020.}
}

\begin{document}

\maketitle



\begin{abstract}
This paper investigates the decoding of a remarkable set of
lattices:
We treat in a unified framework the Leech  lattice in  dimension  24, 
the  Nebe lattice  in dimension 72, and the Barnes-Wall  lattices.
A new interesting lattice is constructed as a simple application of single parity-check principle on the Leech lattice.
The common aspect of these lattices is that they can be obtained via a single parity check or via the $k$-ing construction.
We exploit these constructions to introduce a new efficient paradigm for decoding.
This leads to efficient list decoders and 
quasi-optimal decoders on the Gaussian channel.
Both theoretical and practical performance (point error probability and complexity) of the new decoders are provided. 
\end{abstract}



\section{Introduction}

The Leech  lattice  was discovered at  the dawn  of the
communications  era~\cite{Leech1967}.  Recently,  it  was proved  that
the Leech lattice is  the  densest   packing  of  congruent  spheres  in
24 dimensions~\cite{Cohn2017}. Between these  two major events, it has
been  subject to  countless  studies. This  24-dimensional lattice  is
exceptionally  dense  for its  dimension  and has  a
remarkable  structure.  For  instance, it  contains the  densest known
lattices  in many  lower  dimensions and  it can  be  obtained in
different ways from these lower dimensional lattices. In fact, finding
the simplest structure for efficient decoding of the Leech lattice has
become  a  challenge  among  engineers.  Forney  even  refers  to  the
performance   of    the   best    algorithm   as   a    world   record~\cite{Forney1989}.
Of course, decoding the Leech lattice is not just an
amusing game between engineers as it has many practical interests: Its
high fundamental  coding gain of  6~dB makes it  a good candidate  for high
spectral  efficiency  short  block   length  channel  coding  and  its
spherical-like  Voronoi region  of  16969680 facets~\cite{Conway1999}
enables  to get  state-of-the-art performance  for operations  such as
vector  quantization or  lattice shaping. \\
Recently, Nebe solved a
long standing open problem when  she found an extremal even unimodular
lattice  in dimension~72~\cite{Nebe2012}.
The construction she used to
obtain  this new  lattice  involves  the Leech  lattice  and  Turyn's construction~\cite{Turyn1967} \cite[Chap. 18, Sec~7.4]{MacWilliams1977}.
This 72-dimensional extremal lattice (referred to as the Nebe lattice)  is likely to have better property than  the Leech
lattice for the operations mentioned above.  However, unlike the Leech
lattice, its  decoding aspect has  not been  studied much and,  to the
best  of our  knowledge,  no efficient  decoding  algorithm is known in the literature.
Moreover, none of the existing decoding algorithms for the Leech lattice
seems to scale to the Nebe lattice. The primary motivation of this work
was to propose a new decoder for this lattice\footnote{An efficient decoder was found, see Section~\ref{sec_deco_turyn}.}. 

In this paper, the Leech lattice and the Nebe lattice are presented as special instances of general constructions:
the $k$-ing construction $\ning$
and the single parity-check $k$-lattices $\Gamma(V,\beta,k)_{\mathcal{P}}$, where $ \Gamma(V,\beta,k)_{\mathcal{P}} \subseteq \ning $. 
As examples, the set of lattices obtained as $\ning$ for $k=3$ (known as Turyn's construction \cite[Chap. 18, Sec~7.4]{MacWilliams1977}) include the Leech lattice and the Nebe lattice. 
Regarding the single parity-check $k$-lattices, Barnes-Wall lattices are part of the case $k=2$. 


This framework enables to jointly investigate the construction of several lattices and to present a new decoding paradigm for all of them. 
The paradigm can either be used for bounded-distance decoding (BDD), for list decoding, or for (quasi or exact)-maximum likelihood decoding (MLD) on the additive white Gaussian channel. For regular list decoding (i.e. enumerating all the lattice points in a sphere whose radius is greater than half the minimum distance\footnote{In this paper, we consider squared distances. Therefore, for consistency we should have stated: Greater than a $quarter$ of the minimum distance.} of the lattice), the paradigm can be combined with a technique called the splitting strategy which enables to reduce the complexity. 
Regarding quasi-optimal decoding on the Gaussian channel, our analysis reveals that regular list decoding is not the best choice with our decoding paradigm from a complexity point of view. 
A modified version of the regular list decoder is therefore presented. 
Formulas to predict the performance of these algorithms on the Gaussian channel are provided.


The paper is organized as follows.
Section~\ref{sec_prelem} gives preliminaries.
The $k$-ing construction and the single parity-check $k$-lattices are introduced in Section~\ref{sec_constru_main}.
It is then shown how famous lattices are obtained from these constructions, as well as the parity lattices.
The decoding paradigms are presented in Section~\ref{sec_deco_para}. 
Section~\ref{sec_parity_lat} is dedicated to the study of parity lattices.
Formulas to assess the performance of these algorithms on the Gaussian channel are then provided. 
In Section~\ref{sec_pari_k2}, we further investigate the recursive list-decoding algorithms for the parity lattices with $k=2$ (Barnes-Wall lattices). 
Section~\ref{sec_deco_turyn} focus on the decoding of the Leech and Nebe lattices as special cases of the $k$-ing construction.
Section~\ref{sec_simu_gauss} presents additional numerical results: 
A benchmark of the performance of existing schemes is provided. 
Finally, we draw the conclusions in Section~\ref{sec_conclu} and the appendices are located in Section~\ref{sec_app_main}.



The main contributions of this paper are:

\begin{itemize}
\item A new decoding paradigm to decode $\Gamma(V,\beta,k)_{\mathcal{P}}$ is summarized within Algorithm~\ref{main_alg_pari}. Two list decoding versions of this first algorithm (without and with a technique called the splitting strategy) are then presented.
Moreover, Algorithm~\ref{main_alg_turing} is a direct application of Algorithm~\ref{main_alg_pari} to decode $\ning$. 
See Section~\ref{sec_deco_para}. 
\item
A recursive version of the algorithm of Section~\ref{sec_deco_para} is presented to decode the parity lattices recursively built as $\Gamma(V,\beta,k)_{\mathcal{P}}$.
A modified list-decoding algorithm is proposed for the Gaussian channel. 
Analytic expressions to assess the performance are provided, along with examples.
See Section~\ref{sec_parity_lat}.
\item We show that the parity lattice $L_{3\cdot 24}= \Gamma(V,\beta,3)_{\mathcal{P}}$, as sublattice of $\mathscr{N}_{72}$, has performance only 0.2 dB apart from the Nebe lattice $\mathscr{N}_{72}$ on the Gaussian channel. Moreover, the decoding complexity of $L_{3\cdot 24}$ is significantly reduced. See Section~\ref{sec_parity72}. This is a remarkable result in finding a complexity-performance trade-off. 
\item The case $\Gamma(V,\beta,2)_{\mathcal{P}}$, which includes Barnes-Wall ($BW$) lattices, is also investigated. 
We achieve a lower decoding complexity than the one of existing list decoders for $BW$ lattices.
The modified list-decoding algorithm yields quasi-optimal decoding performance of $BW$ lattices over the Gaussian channel, at a reasonable complexity, up to dimension 128. See Section~\ref{sec_parity_lat}.
\item 
New decoding algorithms for $\Lambda_{24}$ and $\mathscr{N}_{72}$ are developed as an application of our decoding paradigm.
See Section~\ref{sec_deco_turyn}. 
\item These new decoding algorithms uncover the performance of several lattices on the Gaussian channel. 
For instance, Barnes-Wall lattices, the Nebe lattice, and the $3$-parity-Leech lattice $L_{3 \cdot 24}$ are very competitive in their respective dimension: We observe that they have performance similar to known lattices whose dimension is an order of magnitude larger. See Section~\ref{sec_simu_gauss}.
\end{itemize}

\section{Preliminaries}
\label{sec_prelem}

\noindent \textbf{Lattice.} We define a lattice  as a free $J$-module, where
the possible rings $J$ considered in this paper are $\mathbb{Z}$, $\mathbb{Z}[i]$,
and $\mathbb{Z}[\lambda]$, $\lambda = \frac{1+i\sqrt{7}}{2}$.
If $J$ is the regular ring of integers $\mathbb{Z}$, the
lattice $\Lambda$ is a discrete additive subgroup of $\R^n$.
If $J$ is a complex ring of integers, $\Lambda$ is a discrete additive subgroup of $\mathbb{C}^n$ and we say that the lattice is complex.
Given a lattice $\Lambda$ of rank-$n$ in $\R^n$, the rows of a $n\times n$ generator matrix $G$ constitute
a basis of the lattice and any lattice point $x$ is obtained via $x=z \cdot G$, where $z \in \Z^n$.
If it is of rank-$n$ in $\mathbb{C}^n$,
a generator matrix for the corresponding real lattice in $\mathbb{R}^{2n}$ can be obtained as follows.
Map each component $a+ib$ of the complex generator matrix to
\begin{align}
\label{eq_comp_to_real}
\left[
\begin{matrix}
\ \ a &b \\
-b &a
\end{matrix}
\right]
\text{ or }
\left[
\begin{matrix}
\ \ a &b \\
(a-\sqrt{7}b)/2 & (b+\sqrt{7}a)/2
\end{matrix}
\right],
\end{align}
if $J$ is respectively $\mathbb{Z}[i]$ and $\mathbb{Z}[\lambda]$. \\
Given a complex lattice $\Lambda^{\mathbb{C}}$ with generator matrix $G^{\mathbb{C}}$, the lattice generated by
\small
\begin{align}
\label{eq_mat_rot}
\theta \cdot G^{\mathbb{C}}
\end{align}
\normalsize
 is denoted $\theta \Lambda^{\mathbb{C}}$.
Let $\Lambda$, with generator matrix $G$, be the real lattice obtained via \eqref{eq_comp_to_real} 
from the complex lattice $\Lambda^{\mathbb{C}}$.
The real version of $\theta \Lambda^{\mathbb{C}}$, denoted by $\theta \Lambda$, can be either obtained using \eqref{eq_comp_to_real} on
\eqref{eq_mat_rot} or
from $G$ as follows. 
Let $R(2,\theta)$ be the $2\times2$ matrix obtained from $\theta$ via \eqref{eq_comp_to_real}, e.g.
\begin{align}
\small 
R(2, \lambda)=
\left[
\begin{matrix}
1/2   &   \sqrt{7}/2 \\
-\sqrt{7}/2 &  1/2
\end{matrix}
\right] 
\text{ and } 
R(2,\phi)=
\left[
\begin{matrix}
\ \ 1   &   1 \\
-1 &  1
\end{matrix}
\right],
\end{align}
where $\phi = 1+i$.
The scaling-rotation operator $R(n,\theta)$ in dimension $n$ is defined by the application of $R(2,\theta)$ on each pair of components. 
I.e. the scaling-rotation operator is $R(n,\theta)= I_{n/2} \otimes R(2,\theta)$, where $I_{n}$ is the $n\times n$ identity matrix and
$\otimes$ is the Kronecker product.
Then, the real version of $\theta \Lambda^{\mathbb{C}}$ is generated by $G\cdot R(n,\theta)$. \\ 
For a $\mathbb{Z}$-lattice $\Lambda$, the Gram matrix is $G\cdot G^T$. 
The Voronoi cell of $x \in \Lambda$ is:
\begin{equation}
\mathcal{V}(x)=\{ y \in \R^n : \|y-x\| \le \|y-x'\|, \forall x' \in \Lambda \}.
\end{equation}
The fundamental volume of $\Lambda$,
i.e. the volume of its Voronoi cell and its fundamental parallelotope, is denoted by
$\vol(\Lambda)$.
The minimal distance (or minimal squared norm)
of $\Lambda$ is denoted $d(\Lambda)$ and the packing radius is $\rho(\Lambda) = \sqrt{d(\Lambda)}/2$.
We also use $R(\Lambda)$ for the covering radius of $\Lambda$, defined as 
\begin{align}
R(\Lambda)= \underset{y \in \mathbb{R}^n}{\text{max    }} \  \underset{x \in \Lambda}{\text{min}} \ \sqrt{d(y,x)},
\end{align}
where $d(x,y)$ represents the squared Euclidean norm between two elements $x,y \in \R^n$.
The number of lattice points located at a distance $d(\Lambda)$ from the origin is the kissing number $\tau(\Lambda)$.
The fundamental coding gain~$\gamma$ of a lattice $\Lambda$ is given by the following ratio
\vspace{-2mm}
\begin{equation}
\gamma(\Lambda) = \frac{d(\Lambda)}{\text{vol}(\Lambda)^{\frac{2}{n}}}.
\end{equation}
We say that an integral lattice (i.e. the Gram matrix has integer entries) is even
if $\|x\|^2$ is even for any $x$ in $\Lambda$.
Moreover, an integral lattice with $\text{vol}(\Lambda)=1$ is called a unimodular or a self-dual lattice.
Two lattices are equivalent $\Lambda' \cong \Lambda$ if their generator matrices, respectively $G'$ and $G$, are related by $G'=cUGB$,
where $c$ is a non zero constant, $U$ a unimodular matrix, and $B$ an orthogonal matrix. 
If the constant $c$ should be explicit, we write $\Lambda' \cong c\Lambda$. 

Let $\Lambda$ and $\Lambda'$ be lattices where $\Lambda' \subseteq \Lambda$.
If the order of the quotient group $\Lambda/\Lambda'$ is $q$, then $\Lambda$ can be expressed as the union of $q$ cosets of $\Lambda'$.
We denote by $[\Lambda/\Lambda']$ a system of coset representatives for this quotient group. It follows that
\begin{align*}
\Lambda = \bigcup_{x_i \in   [\Lambda/\Lambda']} \Lambda' + x_i = \Lambda'+ [\Lambda/\Lambda'].
\end{align*}
It is simple to prove that \cite[Lem.~1]{Forney1988_1}
\begin{align}
\label{equ_vol}
|\Lambda/\Lambda'| = \frac{\text{vol}(\Lambda')}{\text{vol}(\Lambda)}.
\end{align}

Let $B_r(y)$ be a ball of radius $r$ centered at $y \in \R^n$. The set $\Lambda \cap B_r(y)$, $\Lambda \subset \R^n$, represents the elements $x \in \Lambda$ where $d(x,y) \leq r$.
Let $L(\Lambda,r,y)=|\Lambda \cap B_r(y)|$ be the number of elements in the set $\Lambda \cap B_r(y)$.
The quantity 
\begin{align}
\label{eq_list_nota}
L(\Lambda,r)= \max_{y \in \R^n}|\Lambda \cap B_r(y)|
\end{align}
denotes the maximum number of elements in the set $\Lambda \cap B_r(y)$, for any $y \in \R^n$.
In most situations it will be convenient to consider the relative radius $\delta=r/d(\Lambda)$, which enables to define $l(\Lambda, \delta,y)=L(\Lambda,r,y)$ and $l(\Lambda, \delta)=L(\Lambda,r)$. By abuse of notations, we set  $B_r(y)=B_\delta(y)$; it should be clear from the context whether the radius or relative radius is used.
We also define the relative distance: $\delta(x,y) = \frac{d(x,y)}{d(\Lambda)}$. 

The following Johnson-type bound on the list size for arbitrary
lattices is proved in \cite[Chapter~5]{Micciancio2002}.

\begin{theorem}
\label{lemma_johnson_ori}
Let $\Lambda$ be a lattice in $\R^n$. The list size $L(\Lambda,r)$, defined by \eqref{eq_list_nota}, is bounded as:  
\begin{itemize}
\item $L(\Lambda,r) \leq \frac{1}{2\epsilon }$ if $r \leq d(\Lambda)(1/2 - \epsilon)$, $0<\epsilon \le 1/4$.
\item $L(\Lambda,r) \leq 2n$ if $r \leq d(\Lambda)/2$.
\end{itemize}
\end{theorem}

Let $\Lambda_n \in \R^n$ be part of a family of lattices with instances in several dimensions $n$.
If we want to specify the list size for the lattice in a given dimension $n$,
we simplify the notations as follows. We let $L(n,r)=L(\Lambda_n,r)$ and $l(n,\delta)=l(\Lambda_n,\delta)$. \\

\noindent \textbf{BDD, list decoding, optimal and quasi-optimal decoding (with Gaussian noise).} 
Given a lattice $\Lambda$, a radius $r>0$, and any point $y \in \R^n$,
the task of a list decoder is to determine all points $x \in \Lambda$ satisfying
$d(x,y) \leq r$: i.e. compute the set $\Lambda \cap B_r(y)$.
If $r < \rho^2(\Lambda)$, there is either no point or a unique point found and the decoder is known as BDD. In this paper, BDD means that we consider a decoding radius $r=\rho^2(\Lambda)$ where in case of a tie between several lattice points, one of them is randomly chosen by the decoder.
When $d(x,y) < \rho^2(\Lambda)$, we say that $y$ is within the guaranteed (or unique) error-correction radius of the lattice. 
If $r \ge \rho^2(\Lambda)$, there may be more than one point in the sphere.
In this case, the process is called list decoding rather than BDD.\\
Note that a modified list decoder may output a set of lattice points $\mathcal{T}\neq \Lambda \cap B_r(y)$. 
Therefore, we may refer to list decoders where $\mathcal{T} =\Lambda \cap B_r(y)$ as ``regular'' list decoders. \\
Optimal decoding simply refers to finding the closest lattice point in $\Lambda$ to any point $y \in \R^n$. 
In the literature, it is usually said that an optimal decoder solves the closest vector problem (CVP).
If regular list decoding is used, it is equivalent to choosing a decoding radius equal to $R(\Lambda)$ 
and keeping the closest point to $y$ in the list outputted by the list decoder.

Let $x \in \Lambda$ and $w$ be a Gaussian vector where each component is i.i.d  with distribution $\mathcal{N}(0,\sigma^2)$. Consider the point $y$ obtained as
\small
\begin{align}
\label{eq_gauss_channel}
y = x + w.
\end{align}
\normalsize
Since this model is often used in digital communications, $x$ is referred to as the transmitted point, $y$ the received point, and the model described by \eqref{eq_gauss_channel} is called a Gaussian channel.
The point error probability under optimal decoding is $P_e(opt,\sigma^2) = P(y~\notin~\mathcal{V}(x))$. 
On the Gaussian channel, given equiprobable symbols, optimal decoding is also referred to as maximum likelihood decoding (MLD). 
Moreover, at a fixed dimension $n$, we say that a decoder is quasi-MLD (QMLD) if there exists $\sigma_0^2>0$ and $\epsilon \in (0,1)$ such that  $P_e(dec,\sigma^2) \leq P_e(opt,\sigma^2) \cdot (1 + \epsilon)$. \\
In the scope of (infinite) lattices, the transmitted information rate and the signal-to-noise ratio based on the second-order moment are meaningless. 
Poltyrev introduced the generalized
capacity \cite{Poltyrev1994}, the analog of Shannon
capacity for lattices. The Poltyrev limit corresponds to a noise variance of
$\sigma^2_{max}=\vol(\Lambda)^{\frac{2}{n}}/(2 \pi e)$. 
The point error rate on the Gaussian channel is therefore evaluated 
with respect to the distance to Poltyrev limit, also called the volume-to-noise ratio (VNR), i.e. $\Delta= \sigma^2_{max}/\sigma^2$. \\
The performance of the considered lattices with Gaussian noise, along with their decoders, are compared with the sphere lower bound (see Section~\ref{sec_gauss_BW}). 
For $\mathscr{N}_{72}$, we also plot an approximation of the MLD performance. If the MLD performance is far enough from the Poltyrev limit, the approximation can be based on a truncated \textit{union bound estimate}, which considers only the first lattice shell\footnote{A lattice
shell denotes the set of lattice points at a given distance from
the origin.}. However, as explained in \cite{Forney1998}, this approximation is not accurate if the MLD performance approaches the Poltyrev capacity\footnote{More precisely, since only finite-power constellations are discussed in \cite{Forney1998}, they state that the union bound estimate is not accurate beyond the cutoff rate.}.
Therefore, our estimate is obtained via a union bound computed from the first two shells of the lattice: 
\begin{align}
\label{eq_un_bouns_two_shells}
\tau \cdot Q \left( \sqrt{\frac{d(\Lambda)}{4 \sigma^2}}\right) + \tau' \cdot Q \left( \sqrt{\frac{d(\Lambda)'}{4 \sigma^2}}\right),
\end{align}
where $\tau'$ and $d(\Lambda)'$ are respectively the population and the squared norm of the second lattice shell, and $Q( \cdot )$ is the Gaussian tail function. The dropped terms in the theta series of a lattice \cite{Conway1999} are small o of the first two terms for small $\sigma^2$, so \eqref{eq_un_bouns_two_shells} is tight at high signal-to-noise ratio. 
 \\

\noindent \textbf{Complexity analysis.} 
The complexity of the algorithms is denoted by $\mathfrak{C}$ or $\mathfrak{C}_{A.i}$, where $i$ represents the index of the algorithm. The decoding complexity of a lattice $\Lambda$ is expressed as $\mathfrak{C}(\Lambda)$, where the decoding technique considered is clear from the context.
In general, $\mathfrak{C}$ denotes the worst-case running time. 
By abuse of notation, we use equalities (e.g. $\mathfrak{C}=X$) even though we only provide upper-bounds on the worst-case running time.
We adopt this approach to characterize the complexity of the proposed algorithms, which does not take into account the position of the point to decode $y$.
However, to assess the complexity of the algorithms on the Gaussian channel, 
we take advantage of the distribution of the point $y$ to decode and assess the average complexity $E_y[\mathfrak{C}]$ (warning: $E_y[\mathfrak{C}]$ does not denote the average worst-case complexity but the average complexity).

The complexity of decoding in a lattice $\Lambda$ with a specific decoder is denoted by $\mathfrak{C}^\Lambda_{dec}$, where ``dec" should be replaced by the name of the decoder: E.g. the complexity of BDD, optimal decoding, MLD, and quasi-MLD are $\mathfrak{C}^\Lambda_{BDD}$, $\mathfrak{C}^\Lambda_{opt}$, $\mathfrak{C}^\Lambda_{MLD}$, $\mathfrak{C}^\Lambda_{QMLD}$, respectively.
Moreover, we denote by $\mathfrak{C}_{\Lambda \cap B_{\delta}(y)}$, $\mathfrak{C}_{stor.}^\Lambda$, and $\mathfrak{C}_{clos.}(n)$, the complexity of computing the set $\Lambda\cap B_{\delta}(y)$, storing an element belonging to $\Lambda$, and finding the closest element to $y$ among $n$ elements, respectively. 
If not specified, the set $\Lambda \cap B_{\delta}(y)$ can be computed via the sphere decoding algorithm \cite{Viterbo1999}. 
In this case $\mathfrak{C}_{\Lambda \cap B_{\delta}(y)}=\mathfrak{C}_{Sph. dec.,\delta}^{\Lambda}$.
\\
In general, we assume that $\mathfrak{C}_{dec}^\Lambda >> \mathfrak{C}_{stor.}^\Lambda$ and that $k \mathfrak{C}_{dec}^\Lambda >> \mathfrak{C}_{clos.}(k)$. 
Hence, we have 
\begin{align}
\label{equ_simpli}
k \mathfrak{C}_{dec}^\Lambda + k\mathfrak{C}_{stor.}^\Lambda + \mathfrak{C}_{clos.}(k) \approx k\mathfrak{C}_{dec}^\Lambda.
\end{align}
Similarly, we also have:
\begin{align}
\label{equ_simpli_2}
\mathfrak{C}_{\Lambda\cap B_{\delta}(y)}+ l(\Lambda,\delta) \mathfrak{C}_{stor.}^\Lambda + \mathfrak{C}_{clos.}(l(\Lambda,\delta)) \approx \mathfrak{C}_{\Lambda\cap B_{\delta}(y)}.
\end{align}
By abuse of notations, we may write $k \mathfrak{C}_{dec}^\Lambda + k\mathfrak{C}_{stor.}^\Lambda + \mathfrak{C}_{clos.}(k) = k\mathfrak{C}_{dec}^\Lambda$ (e.g. if $\Lambda \in \R^{\frac{n}{k}}$, we sometimes write $k \mathfrak{C}_{dec}^\Lambda + O(n)=k \mathfrak{C}_{dec}^\Lambda$ if the $O(n)$ is not relevant in the context). 
When recursively decoding a lattice $\Lambda_n \in \R^n$, we simplify the notation $\mathfrak{C}(\Lambda_n,\delta)$ by $\mathfrak{C}(n,\delta)$. 

The $\widetilde{O}$ notations is used to ignore the logarithmic factors. The notation $f(n) = \widetilde{O}(h(n))$ is equivalent to $\exists k$ such that $f(n) = O(h(n) \log^k (h(n)))$ (since $\log^k(n)$ is always $o(n^\epsilon)$ for any $\epsilon>0)$. \\

\noindent \textbf{Extremal lattice.} The fundamental coding gain of an even unimodular lattice of dimension $n$ is at most $2 \lfloor \frac{n}{24} \rfloor +2$. Lattices achieving this coding gain are called extremal. \\
%

\section{Lattice construction}
\label{sec_constru_main}
\subsection{The $k$-ing construction and the single parity-check lattice}
\label{sec_king}
Consider lattices $S,T,V$,  where $V \subset T \subset S$.
Let us denote $\alpha=[S/T]$ and $\beta=[T/V]$, two groups of coset representatives.
The $k$-ing construction of a lattice $\Gamma$ is defined as
\small
\begin{flalign}
\label{equ_ning_2}
\begin{split}
 \ning  
		 = &\{ (m+t_1, m+ t_2,  ... , m+t_k),  m \in \underset{T^*}{\underbrace{V+\alpha}}, t_i \in \underset{T}{\underbrace{V+\beta}}, \sum_{i=1}^{k} t_i \in V \} \subseteq S^k,
\end{split}
\end{flalign}
\normalsize
since $V+\beta$ is the lattice $T$. 
We denote the lattice $V+\alpha$ by $T^*$. $\ning$ can alternatively be defined from $T$ and $T^*$, it is also denoted by  $ \Gamma(V,T^*, T,k)$.

An obvious sublattice of $\ning$ is the single parity-check lattice in $T^k$:
%

\small
\begin{flalign}
\label{equ_pari_constA}
\begin{split}
\Gamma(V,\beta,k)_{\mathcal{P}} = \Gamma(V,T,k)_{\mathcal{P}} =& \{ (t_1, t_2,  ... , t_k) \in T^k | \sum_{i=1}^k t_i \in V \} = \{ (t_1, t_2,  ... , v_{k} - \sum_{i \neq k } t_i),   v_k \in V, \  t_i \in T \}, \\
\end{split}
\end{flalign}
\normalsize
where the last expression is the most useful in practice.

Using $\Gamma(V,\beta,k)_{\mathcal{P}}$, the lattice $\ning$ can be represented as follows, an expression useful for decoding. 
\small
\begin{flalign}
\label{equ_ning_coset_m}
\begin{split}
\ning &= \bigcup_{m \in \alpha}  \{ \Gamma(V,\beta,k)_{\mathcal{P}} + m^k \}.
\end{split}
\end{flalign}
\normalsize
The set $V+\alpha$ for $m$ in \eqref{equ_ning_2} becomes $\alpha$ in \eqref{equ_ning_coset_m} after moving the $V$ components into the $t_i's$.
From \eqref{equ_pari_constA}, 
we easily see that $d(\Gamma(V,\beta,k)_{\mathcal{P}})=\min\{d(V), 2d(T)\}$. The minimum distance of  $\ning$ is provided by the next theorem, proved in \cite{Forney1988} \cite{Desideri2003}.
\normalsize
\begin{theorem}
\label{theo_min_dist}
The minimum distance of $\ning$ satisfies 
\small
\begin{align}
\begin{split}
\min \{d(V),2d(T)\} \ge d(\ning) \ge  \min \{d(V), 2d(T), kd(S)\}.
\end{split}
\end{align}
\normalsize
\end{theorem}


Families of single parity-check lattices can be built by recursively applying the single parity-check construction (Equation~\eqref{equ_pari_constA}).
For instance, a new family of lattices is obtained as follows. 
First, since $d(\Gamma(V,\beta,k)_{\mathcal{P}})=\min\{2d(T),d(V)\}$, we shall consider only lattices where $d(V)=2d(T)$. 
In order to find two lattices having this property, with $V\subset T$, we consider a lattice $T$ over a complex ring $J$, and rotate it by an element $\theta \in J$, with $|\theta| = \sqrt{2}$, to get $V$: i.e. $V=\theta \cdot T$. This yields $V \cong \sqrt{2}T$ and $d(\Gamma(\theta T,\beta,k)_{\mathcal{P}})=2d(T)$. The ring $J$ can for instance be $\mathbb{Z}[i]$ or $\mathbb{Z}[\lambda]$.
More formally, let $\Lambda_c^{\mathbb{C}} \in \mathbb{C}^{c/2}$ be a lattice over a complex ring $J$, where $J$ is either $\mathbb{Z}[i]$ or $\mathbb{Z}[\lambda]$.
We denote by $\Lambda_c~\in~\mathbb{R}^c$ the corresponding real lattice, with real dimension $c$. 
Let $L_{n}$ be the real lattice obtained via \eqref{eq_comp_to_real} 
from a complex lattice $L_{n}^{\mathbb{C}}$.
In the sequel, $\theta L_{n}$ is the notation for the real lattice obtained from $\theta \cdot L_{n}^{\mathbb{C}}$.
Also, $\beta = [L_{n}/\theta L_{n}]$. \\
\begin{definition}
\label{def_pari_lat}
Let $n=c \cdot k^t$, $t \ge 0$.
The parity lattices in dimension $kn$ are defined by the following recursion:
\small
\begin{align}
\begin{split}
L_{kn} &= \Gamma(\theta L_{n},\beta,k)_{\mathcal{P}}, \\
				  & = \{ (v_1+n_1,v_2 + n_2,...,v_k -  \sum_{i \neq k} n_i),  v_i \in \theta L_{n}, n_i \in \beta \}, \\
			&= \{ (t_1,t_2,...,v'_k -  \sum_{i \neq k} t_i),  v'_k \in \theta L_{n}, t_i \in L_{n} \}, \\
\end{split}
\end{align}
\normalsize
with initial condition $L_{c} =\Lambda_{c}$. The number of recursive steps $t$ should not be confused with $t_i \in L_n$. 
\end{definition}

As we shall see in the sequel, several famous lattices can be obtained via the $k$-ing construction or as parity lattices, including the Leech lattice and the Nebe lattice as well as the Barnes-Wall lattices.

Most papers study the $k$-ing construction for a fixed $k$. As a result, various names exist for this construction given a fixed~$k$:
If $k=3$, it is Turyn's construction \cite[Chap. 18, sec 7.4]{MacWilliams1977} (also known as the cubing construction \cite{Forney1988}). If $k=4$, it is the two-level squaring construction \cite{Forney1988}. If $k=5$, it is the quinting construction \cite{Ling2001}. If $k=7$, it is the septing construction \cite{Ling2001}. 

\noindent \textbf{Example.} 
Take $V=2 \mathbb{Z}$ and $T=\mathbb{Z}$. Then, $\beta=[\mathbb{Z}/2\mathbb{Z}]$. The checkerboard lattice is $D_n=\Gamma(V,\beta,n)_{\mathcal{P}}$.

\subsection{Parity lattices with $k=2$ ($BW$ lattices)}

We consider the parity lattices obtained with $k=2$, $L_c=\mathbb{Z}^2$ and $\theta= \phi=1+i$. 
They are called Barnes-Wall lattices in the literature \cite{Conway1999}. These lattices are recursively expressed as
\small
\begin{align}
 BW_{2n}= \Gamma(\phi BW_n,\beta,2)_{\mathcal{P}}  \text{ where } BW_2=\mathbb{Z}^2. 
\end{align}
\normalsize
In general, the lattice $\phi BW_n$ is denoted by $RBW_n$ \cite{Forney1988}. 
We adopt this latter notation for the rest of the paper. 
These BW lattices were one of the first series discovered
with an infinitely increasing fundamental coding gain~\cite{Barnes1959}: It increases as $\gamma(BW_{n}) = \sqrt{2} \cdot \gamma(BW_{n/2}) = \sqrt{n/2}$. 

\subsection{Leech lattice and Nebe lattice}
\label{sec_turyn}

Both Leech and Nebe lattices can be obtained via Turyn's construction, i.e. as $\Gamma(V,\alpha,\beta,3)$.


Among the three lattices $S,T,V$ used in the construction $\Gamma(V,\alpha,\beta,3)$, let us take $V = 2S$.
To build the Leech lattice, we have $S,T \cong E_8$ and to build the Nebe lattice we have 
$S,T \cong \Lambda_{24}$.
Moreover, to obtain these two lattices via the $k$-ing construction, the set of coset representatives $\alpha$ should be chosen such that $d(\Gamma(V,\alpha,\beta,3))>3d(S)$ (instead of $\ge$ as in Theorem~\ref{theo_min_dist}).
We already established via~\eqref{equ_ning_2} that choosing $\alpha$ is equivalent to choosing $T^*$.
In the next section, we explain how to get $T^*$ via lattice polarisation \cite{Nebe2012}.


\subsubsection{The polarisation of lattices}
\label{sec_polar}

%
Assume that the lattices $S,T,T^*,V=2S$ are of rank $n$.
Here, $T^*$ is a rotation
of $T$ by an angle of $2 \theta$. Therefore, it is denoted $T_{2 \theta}$.

\begin{definition}
Given a lattice $S$, we call $(T,T_{2 \theta})$ a polarisation of $S$ \cite{Nebe2012} if 
\small
\begin{align}
\label{eq_prop}
\begin{split}
&S\cong T \cong T_{2 \theta},  S=T_{2 \theta}+T, ~~~\text{and} ~~~T_{2 \theta} \bigcap T =2S.
\end{split}
\end{align}
\normalsize
\end{definition}

Let $G_{S}$ be a generator matrix of $S$.
Finding a polarisation of the lattice $S$ (if it exists) is equivalent to finding a scaling-rotation matrix $R$, $R\cdot R^{T} = 2I$, with
$G_{T}=G_{S} \cdot R  \text{   and   } G_{T_{2 \theta}}=G_{S} \cdot R^{T}$,
such that the basis vectors $g_{T}^{i}$ and $g_{T_{2 \theta}}^{i}$, $1\le i \le n$, 
are versions of the vectors $g_S^i$ scaled by a factor of $\sqrt{2}$  and rotated by an angle of $\pm \theta = \text{arctan}\sqrt{7}$. 
Indeed, consider two vectors $g_T^i$ and $g_{T_{2\theta}}^i$ of the same size and having an angle of $2\theta$. Summing these two vectors yields a vector $g_S^i$ having half the size of $g_T^i$: 
$||g_{S}^i||^2 = ||g_T^i+g_{T_{2\theta}}^i||^{2}= 0.5 \times ||g_{T}^i||^2$,
since $\cos \left( 2 \theta \right)=-3/4$. One would thus get $G_S = G_T + G_{T_{2\theta}}$.
The rotation matrix $R$ can be found via a $\mathbb{Z}[\lambda]$-structure of $S$;
Let $G_{S}^{\mathbb{C}}$  be a (complex) generator matrix of $S$ over the ring of integers $\mathbb{Z}[\lambda]$, $\lambda~=~\sqrt{2}e^{i \theta}= \frac{1 + i\sqrt{7}}{2}$.
Multiplying $G_{S}^{\mathbb{C}}$ by $\lambda$ yields a matrix whose rows are new vectors belonging to the lattice, scaled by $\sqrt{2}$, and having the desired angle with the basis vectors.
Hence, $G_T^{\mathbb{C}}$ can be obtained as $\lambda G_{S}^{\mathbb{C}}$ and $G_{T_{2 \theta}}^{\mathbb{C}}$ as $\psi G_{S}^{\mathbb{C}}$, where $\psi = \bar{\lambda}$ is the conjugate of $\lambda$.
Therefore, if we let $G_{S}$ be the real generator matrix obtained from $G_{S}^{\mathbb{C}}$ (via \eqref{eq_comp_to_real}), the real rotation matrix $R$ for polarisation is
$R(n,\lambda)=I_{n/2} \otimes R(2,\lambda).$

\subsubsection{Leech lattice and Nebe lattice}
\label{sec_constru_leech}

Given three lattices $S$, $T$ and $T_{2\theta}$, respecting properties $\eqref{eq_prop}$, we consider the lattice $\turing$. 
A generator matrix of $\turing$ is
\small
\begin{align}
\label{eq_Forney}
G_{\turing} & = 
\left[
\begin{matrix}
G_{(3,1)} \otimes G_{T_{2\theta}}\\
G_{(3,2)} \otimes G_{T} 
\end{matrix} 
\right]
\end{align}
\normalsize
where $G_{(3,1)}$ and $G_{(3,2)}$ are generator matrices for the $(3,1)$ binary repetition code and the $(3,2)$ binary single parity-check code, respectively. Obviously, $\mathbb{F}_2$ is naturally embedded into $\mathbb{Z}$ for the two binary codes.
From~\eqref{eq_Forney}, a generator matrix of $\turing$ over $\mathbb{Z[\lambda]}$ can be expressed as
\small
\begin{align}
\label{eq_complex_pb}
G^{\mathbb{C}}_{\turing} =
\underset{=Pb}{
\underbrace{
 \left[
\begin{matrix}
\lambda & \lambda & \lambda \\
\psi&  \psi& 0 \\
0&  \psi& \psi
\end{matrix}
\right]
}
}
 \otimes
G^{\mathbb{C}}_{S},
\end{align}
\normalsize
where $G^{\mathbb{C}}_{S}$ is a generator matrix of $S$ over $\mathbb{Z[\lambda]}$, $\lambda G^{\mathbb{C}}_{S}$ a generator matrix of $T$, and $\psi G^{\mathbb{C}}_{S}$ a generator matrix of $T_{2 \theta}$.

\begin{theorem}\label{th_turyn}
Let $S \cong E_8$ and $T$,$T_{2\theta}$ be two lattices respecting
properties $\eqref{eq_prop}$.
Then, $\turing$ is the Leech lattice \cite{Tits1978}\cite{Lepowsky1982}\cite{Quebbemann1984}. 
\end{theorem}
See Appendix~\ref{App_proof_Leech_turing} for a proof.
A similar proof can also be used to show that when $S \cong \Lambda_{24}$ and $T$,$T_{2\theta}$ are two lattices respecting
properties $\eqref{eq_prop}$, then the lattice $\turing$ has a fundamental coding gain equal to 6 or 8 \cite{Griess2010}.
In this case, the polarisation does not ensure $\turing > 3d(S)=6$. Additional work to choose
$T_{2 \theta}$ is needed. Nebe considers in \cite{Nebe2012} the following construction.
Let $S$ be the $\mathbb{Z}[\lambda]$-structure $\Lambda_{24}$ with automorphism group SL$_{2}(25)$. Set $T_{2\theta} = \lambda S $ and $T = \psi S$. 
The resulting lattice $\turing$ is named the Nebe lattice $\mathscr{N}_{72}$.
It is possible to check that $\mathscr{N}_{72}$ has no vector of length 6, which leads to the following theorem, obtained in \cite{Nebe2012}.





%
\begin{theorem}\label{turyn_nebe2}
$\mathscr{N}_{72}$  has a fundamental coding gain equal to $8$ \cite{Nebe2012}.
\end{theorem}

%




\section{Decoding paradigm for the single parity-check lattice and the $k$-ing lattice}
\label{sec_deco_para}


\subsection{The existing decoding algorithm for $\Gamma(V,\alpha,\beta,k)$ (and $\Gamma(V,\beta,k)_{\mathcal{P}}$)}
\label{sec_trellis_dec}

To the best of the authors' knowledge, there exists only one ``efficient" optimal algorithm for the $k$-ing construction called trellis decoding \cite{Forney1988}.
This decoding algorithm uses a graph-based representation to efficiently explore all the cosets of $V^k$ in $\ning$.
As an example, the trellis for $\Gamma(V,\alpha,\beta,3)$ is illustrated on Figure~\ref{fig_polar} with $|\alpha|=3$ and $|\beta|=2$. Each path in this three sections trellis corresponds to a coset of $V^3$ in $\Gamma(V,\alpha,\beta,3)$. Each edge is associated with a coset of $V$ in $S$: E.g. given $\alpha = \{m_1,m_2,m_3\}$ and $\beta=\{n_1,n_2\}$, the two upper edges on the left should be labeled $m_1+n_1$ and $m_1+n_2$, respectively. All the edges in the upper part of the trellis correspond to the same $m_1$ and the sub-trellis formed by these edges is a standard single parity-check trellis. This sub-trellis is repeated three times for $m_1$,$m_2$, and $m_3$. Standard trellis algorithms, such as the Viterbi algorithm, can then be used to decode.

\begin{figure}[H]
\centering
\includegraphics[width=0.25\columnwidth]{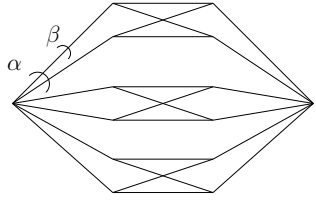}
\caption{Trellis representing a $\Gamma(V,\alpha,\beta,3)$ with $|\alpha|=3$ and $|\beta|=2$. The edges labeled with the same $\alpha$ are associated with the same $m_1 \in \alpha$.}
\label{fig_polar}
\end{figure}
Trellis decoding of $\Gamma(V,\alpha,\beta,k)$ involves decoding in $V$ for each edge in the trellis. The number of edges in the trellis is $2  |\alpha| |\beta| + (k-2)  |\alpha||\beta|^2$.
Therefore, the complexity is dominated by the quantity $ |\alpha|| \beta|^2 \mathfrak{C}^{V}_{dec}$.
For more information on trellis decoding, the reader should refer to \cite{Forney1988} or \cite{Desideri2008}.

Of course, trellis decoding can also be used to decode the single parity-check $k$-lattices $\Gamma(V,\beta,k)_{\mathcal{P}}$. 
The number of edges in the standard single parity-check trellis is $2 |\beta| + (k-2) |\beta|^2$. 

\subsection{Decoding paradigm for $\Gamma(V,\beta,k)_{\mathcal{P}}$}
\label{sec_main_deco}


Set $y\in \R^{kn}$ and let $x=(t_1,t_2, ... , t_k) \in \Gamma(V,\beta,k)_{\mathcal{P}}$ be the closest element to $y$. 
The minimum distance of the lattice $V$ is (in general) larger than the one of $T$. Consequently, decoding $y_j$ in the coset of $V$ to which the element $t_j$ belongs is safer than decoding in $T$.
Moreover, any set of $k-1$ $t_j$'s is enough to know in which coset of $V^k$ in $\Gamma(V,\beta,k)_{\mathcal{P}}$ is located $x$. 
Hence, given $t_1,t_2,...t_{k-1}$, the element $t_k$ can be recovered by decoding $y_k-(-\sum_{j=1}^{k-1}t_j)$ in $V$ (and adding back $-\sum_{j=1}^{k-1}t_j$ on the decoded element), as shown by Algorithm~\ref{main_alg_pari}. 
It is easily seen that the complexity of Algorithm~\ref{main_alg_pari} is 
\begin{align}
\label{equ_complex_alg1}
\mathfrak{C}_{A.\ref{main_alg_pari}}= k \mathfrak{C}^T_{dec} + k \mathfrak{C}^V_{dec}, 
\end{align}
where we used the simplification of Equation~\eqref{equ_simpli}.
\begin{algorithm}
\caption{Decoder for $\Gamma(V,\beta,k)_{\mathcal{P}}$}
\label{main_alg_pari}
\textbf{Input:} $y=(y_1,y_2, ... , y_k) \in \mathbb{R}^{kn}$.   
\begin{algorithmic}[1]
\STATE Decode $y_1,y_2,...,y_k$ in $T$ as $t_1,t_2,...,t_k$.
 \FOR{$1 \leq i \leq k$} 
\STATE {
Decode $y_i-(- \sum_{j \neq i}t_j)$ in $V$ as $v_i$. 
Add $(t_1,...,t_{i-1},v_i+(-\sum_{j \neq i}t_j),t_{i+1}, ..., t_k)$ to the list $\mathcal{T}$.
} 
\ENDFOR
\STATE \textbf{Return} the closest element of $\mathcal{T}$ to $y$.
\end{algorithmic}
\end{algorithm}

Algorithm~\ref{main_alg_pari} can be adapted to perform list decoding as follows. 
For the sake of simplicity, 
we assume that $V \cong \sqrt{2}T$. 
We recall that regular list decoding consists in computing the set $\Gamma(V,\beta,k)_{\mathcal{P}} \cap B_{\delta}(y)$:
i.e. finding all lattice points $x \in \Gamma(V,\beta,k)_{\mathcal{P}}$ where $d(y,x)\leq r$, $y \in \R^{kn}$. 
The parameter $\delta=r/d(\Gamma(V,\beta,k)_{\mathcal{P}})$ is the relative decoding radius. 
Remember that $d(\Gamma(V,\beta,k)_{\mathcal{P}})=d(V)=2d(T)$. 
The list decoding of $\Gamma(V,\beta,k)_{\mathcal{P}}$ with a radius $r$ 
consists in list decoding each $y_j$  (Step~1 of Algorithm~\ref{main_alg_pari}) in $T$ with a radius $r/2$ and each $y_i-(-\sum_{j\neq i} t_j)$ (Step~3) in $V$ with a radius $r$ (see the proof of Lemma~\ref{theo_worst} for explanations on this choice). 
In both cases the relative radius is $\delta = \frac{r}{2d(T)}=\frac{r}{d(V)}$ and the maximum number of elements in each list is $l(T,\delta)=l(V,\delta)=L(T,\frac{r}{2})=L(V,r)$  (see Section~\ref{sec_prelem} for the definitions of $L(\cdot,\cdot)$ and $l(\cdot,\cdot)$). 
As a result, Step~3 (of Algorithm~\ref{main_alg_pari}), for a given $i$, should be executed for any of the combinations of candidates (for each $t_{j\neq i}$) in the $k-1$ lists: i.e. $l(T,\delta)^{k-1}$ times. The resulting maximum number of stored elements (for this given $i$) is $l(T,\delta)^{k-1}\cdot l(V,\delta)$. 
Consequently, the number of elements in $\mathcal{T}$ is bounded from above by 
\begin{align}
k \cdot l(T,\delta)^{k-1}\cdot l(V,\delta) =k \cdot  l(T,\delta)^{k}.
\end{align}
The list-decoding version of Algorithm~\ref{main_alg_pari} is presented in Algorithm~\ref{main_alg_pari_no_sp}. 
\begin{lemma}
\label{theo_worst}
Algorithm~\ref{main_alg_pari_no_sp} outputs the set $\Gamma(V,\beta,k)_{\mathcal{P}} \cap B_{\delta}(y)$ in worst-case time
\begin{align}
\label{equ_complex_list_without}
\small
\begin{split}
\mathfrak{C}_{A.\ref{main_alg_pari_no_sp}} = &
k \mathfrak{C}_{T\cap B_{\delta}(y)} +  k \cdot l(T,\delta)^{k-1}  \mathfrak{C}_{V\cap B_{\delta}(y)}.
\end{split}
\end{align}
\end{lemma}
\begin{proof}
We first prove that all points  $x=(x_1,x_2,...,x_k) \in \Gamma(V,\beta,k)_{\mathcal{P}} \cap B_{\delta}(y)$ are outputted by Algorithm~\ref{main_alg_pari_no_sp}. 
If $d(y_i,x_i)>r/2$ then $d(y_j,x_j)<r/2$ for all $j\neq i$.
Hence, among the $k$ lists $\mathcal{T}_1,\mathcal{T}_2,..., \mathcal{T}_k$ computed at Step~1 of the algorithm, 
at least $k-1$ of them contain the correct $t^*_j=x_j$. \\
Assume (without loss of generality) that all $\mathcal{T}_j$, $1\leq j\neq i \leq n$, contain $t^*_j$. 
Since $d(y,x) \leq r$, one has $d(y_i,x_i) \leq r$. 
Therefore, $\mathcal{V}_i = V \cap B_{\delta}(y_i-(-\sum_{j \neq i} t^*_j))$ contains $v^*_i=x_i-(-\sum_{j \neq i} t^*_j)$. \\
As a result, all $x \in \Gamma(V,\beta,k)_{\mathcal{P}} \cap B_{\delta}(y)$ are outputted by the algorithm. 
The complexity is obtained by reading Algorithm~\ref{main_alg_pari_no_sp},
with the simplification of Equation~\eqref{equ_simpli}.
\end{proof}
Note that if $\delta<\frac{1}{4}$ (the relative packing radius), there is only one element in each of the set $\mathcal{T}_j$ computed in Algorithm~\ref{main_alg_pari_no_sp}. Algorithm~\ref{main_alg_pari_no_sp} in this case is equivalent to Algorithm~\ref{main_alg_pari}. 

%
%
\begin{algorithm}
\caption{List dec. for $\Gamma(V,\beta,k)_{\mathcal{P}} \in \R^{kn}$ (without the splitting strategy) }
\label{main_alg_pari_no_sp}
\textbf{Input:} $y=(y_1,y_2, ... , y_k) \in \R^{kn}$, $\delta \geq 0$.   
\begin{algorithmic}[1]
\STATE Compute the sets $\mathcal{T}_1 = T \cap B_\delta(y_1), \mathcal{T}_2 = T \cap B_\delta(y_2), ..., \mathcal{T}_k= T \cap B_\delta(y_k)$.
 \FOR{$1 \leq i \leq k$} 
\STATE Set $j_1<j_2<...<j_{k-1}$, where $\{ j_1,j_2,...,j_{k-1} \} = \{1,2,...,k\} \backslash \{i\}$.
 \FOR{each $(t_{j_1},...,t_{j_{k-1}}) \in \mathcal{T}_{j_1} \times \mathcal{T}_{j_2}\times ... \times \mathcal{T}_{j_{k-1}}$}
\STATE {
Compute the set $\mathcal{V}_i = V \cap B_{\delta}(y_i-(- \sum_{j'=1}^{k-1}t_{j_{j'}}))$.
} 
\FOR{$v_i \in \mathcal{V}_i$}
\STATE Add $(t_{j_1},...,t_{j_{i-1}},v_i+(-\sum_{j'=1}^{k-1}t_{j_{j'}}), t_{j_{i}}, ..., t_{j_{k-1}})$ to the list $\mathcal{T}$.
\ENDFOR
\ENDFOR
\ENDFOR
\STATE \textbf{Return} $\mathcal{T}$.
\end{algorithmic}
\end{algorithm}

In some cases, the complexity of Algorithm~\ref{main_alg_pari_no_sp} can be reduced via a technique we call the splitting strategy. 
It exploits the following observation: Let $x=(x_1,x_2,...,x_k) \in \Gamma(V,\beta,k)_{\mathcal{P}}$.
Assume (without loss of generality) that $d(y_i,x_i)> \frac{r}{2}$ (and thus $\sum_{j\neq i} d(x_j,y_j) \leq \frac{r}{2}$). This case can be split into two sub-cases. 
Let $0 \leq a' \leq \frac{r}{2}$.
\begin{itemize}
\item If $a' \le \sum_{j\neq i} d(x_j,y_j) \leq \frac{r}{2}$ then $\frac{r}{2}  < d(x_i,y_i)\le r-a'$: \\
Then, each $y_j$ should be list decoded in $T$ with a radius $\frac{r}{2}$ 
and $y_i-(-\sum_{j\neq i} t_j)$, for all resulting combinations of $t_j$, list decoded in $V$ with a radius $r-a'$. 
\item Else $0 \leq \sum_{j\neq i} d(x_j,y_j)  < a' $ and $\frac{r}{2}  < d(x_i,y_i) \le r$: \\
Then, each $y_j$ should be list decoded in $T$ with a radius $a'$, 
and $y_i-(-\sum_{j\neq i} t_j)$, for all resulting combinations of $t_j$, list decoded in $V$ with a radius~$r$. 
\end{itemize}
The number of stored elements (when computing each sub-case) is bounded by 
\begin{itemize}
\item $l(T,\delta)^{k-1}\cdot l(V,a_1=\frac{r-a'}{d(V)})$, for the first sub-case,
\item $l(T,a_2=\frac{a'}{d(T)})^{k-1}\cdot l(V,\delta)$ for the second sub-case.
\end{itemize}
Consequently, if we choose $a_1=a_2=\frac{2}{3}\delta$, the number of elements in the list $\mathcal{T}$, outputted by a list decoder with this splitting strategy, is bounded from above by
\small
\begin{align}
\label{equ_split_first}
k \big[ l(T,\delta)^{k-1}  l(V,\frac{2}{3}\delta) + l(T,\frac{2}{3}\delta)^{k-1} l(V,\delta) \big],
\end{align} 
\normalsize
which is likely to be smaller than $k \cdot  l(T,\delta)^{k}$, the bound obtained without the splitting strategy.  

Similarly, we can also split the case $0 \le d(x_j,y_j) \leq \frac{r}{2}$, $j \neq i$, into several sub-cases.
Let $0 \leq a' \leq \frac{r}{2}$. We recall that we have $\sum_{j\neq i} d(x_j,y_j) \leq \frac{r}{2}$.
\begin{itemize}
\item If $a' \le d(x_j,y_j) \leq \frac{r}{2}$ then $0 \le d(x_l,y_l)\le \frac{r}{2}-a'$, $\forall l$ where $1 \leq l \neq j \neq i \leq k$: \\
Then, $y_j$ should be list decoded in $T$ with a radius $\frac{r}{2}$ 
and each $y_l$ list decoded in $T$ with a radius $\frac{r}{2}-a'$. 
\item Else $0 \le d(x_j,y_j) < a'$ and for one $l$, $1 \leq l \neq j \neq i \leq k$, one may have  $a' \le d(x_l,y_l)\le \frac{r}{2}$: \\
Then, $y_j$ should be list decoded in $T$ with a radius $a'$, 
$y_l$ list decoded in $T$ with a radius $\frac{r}{2}$, and all the remaining $y's$ 
list decoded in $T$ with a radius $a'$. 
\end{itemize} 
Of course, since it is not possible to know the index $l$ where\footnote{One may not have $a' \leq d(x_l,y_l)$, $\forall l$, $1\leq l \neq i \leq k$. It is not an issue as we would then simply decode with a radius greater than necessary.} $a' \leq d(x_l,y_l)\le \frac{r}{2}$, all $k-2$ possibilities should be computed (which yields $k-1$ possibilities if we include the first sub-case $a' \le d(x_j,y_j) \leq \frac{r}{2}$).
If we choose $a'=\frac{r}{2}-a'=\frac{r}{4}$, the product of the maximum list size of each $k-1$ case is $l(T,\delta)l(T,\frac{\delta}{2})^{k-2}$.
As a result, the maximum number of possibilities to consider for $\sum_{j\neq i}t_j$ is 
\small
\begin{align}
\label{eq_sec_split}
(k-1)l(T,\delta)l(T,\frac{\delta}{2})^{k-2},
\end{align}
\normalsize
instead of $l(T,\delta)^{k-1}$ without this strategy. 

Substituting \eqref{eq_sec_split} in \eqref{equ_split_first},
the number of element in a list $\mathcal{T}$, outtputed by a list decoder with these two splitting strategies, is bounded from above by  
\begin{align}
\label{eq_complex}
\begin{split}
\small
\begin{split}
&k(k-1)\big[ l(T,\delta)l(T,\frac{\delta}{2})^{k-2}  l(V,\frac{2}{3}\delta)+ l(T,\frac{2}{3}\delta)l(T,\frac{\delta}{3})^{k-2} l(V,\delta) \big], \\
&=k(k-1) l(T,\delta)  l(T,\frac{2}{3} \delta) \big[l(T,\frac{\delta}{2})^{k-2} + l(T, \frac{\delta}{3})^{k-2}\big],
\end{split}
\end{split}
\end{align}
where we used $l(T,\delta)=l(V,\delta)$.

We shall refer to these two splitting strategies as the first and second splitting strategy, respectively.
The first splitting strategy is listed in Algorithm~\ref{main_alg_pari_split} and can be used without or with the second splitting strategy. The function $SubR_1$ or $SubR_2$, listed in Algorithm~\ref{algo_subr_m_0}, is used accordingly.

\begin{lemma}[Complexity with the splitting strategy]
\label{theo_set_split}
Algorithm~\ref{main_alg_pari_split}, with the subroutine $SubR_2$ listed in Algorithm~\ref{algo_subr_m_0}, outputs the set $\Gamma(V,\beta,k)_{\mathcal{P}} \cap B_{\delta}(y)$ in worst-case time
\begin{align}
\label{equ_complex_big_algo_2}
\small
\begin{split}
\mathfrak{C}_{A.\ref{main_alg_pari_split}} =&k \mathfrak{C}_{T\cap B_{\delta}(y)} + (k^2-k) \Big[l(T,\delta) l(T,\frac{\delta}{2})^{k-2} \mathfrak{C}_{V\cap B_{\frac{2}{3}\delta}(y)}  +  l(T,\frac{2}{3} \delta) l(T,\frac{\delta}{3})^{k-2} \mathfrak{C}_{V\cap B_{\delta}(y)} \Big] .
\end{split}
\end{align}
\end{lemma}

\begin{algorithm}
\caption{List dec. for $\Gamma(V,\beta,k)_{\mathcal{P}} \in \R^{kn}$ with the splitting strategy}
\label{main_alg_pari_split}
\textbf{Input:} $y=(y_1,y_2, ... , y_k) \in \R^{kn}$, $\delta \geq 0$.  \\
// The sets $\mathcal{T}^{\frac{\delta}{2}}_i$ and $\mathcal{T}^{\frac{\delta}{3}}_i$ are computed only if used by the subroutine.
\begin{algorithmic}[1]
\FOR{$\eta \in \{\delta, \frac{2}{3} \delta, \frac{\delta}{2}, \frac{\delta}{3} \}$} 
\STATE Set $\mathcal{T}_1^\eta, \mathcal{T}_2^\eta,...,\mathcal{T}_k^\eta $ as global variables.
\STATE Compute the sets $\mathcal{T}_1^\eta = T \cap B_\eta(y_1), \mathcal{T}_2^\eta = T \cap B_\eta(y_2), ..., \mathcal{T}_k^\eta= T \cap B_\eta(y_k)$.
\ENDFOR
 \FOR{$1 \leq i \leq k$}

\STATE $\mathcal{T}_1 \leftarrow SubR(y_1,y_2,...,y_k,\delta,\frac{2}{3} \delta, i)$. 
\STATE $\mathcal{T}_2 \leftarrow SubR(y_1,y_2,...,y_k,\frac{2}{3}\delta,\delta, i)$.
\\
// Use $SubR_1$ or $SubR_2$ (listed in Algorithm~\ref{algo_subr_m_0}) if the (second) splitting strategy is used or not, respectively. 

\ENDFOR
\STATE \textbf{Return}  $\mathcal{T}=\{\mathcal{T}_1,\mathcal{T}_2\}$. \\ 

\end{algorithmic}
\end{algorithm}

\begin{algorithm}
\caption{Subroutines of Algorithm~\ref{main_alg_pari_split}}
\label{algo_subr_m_0}
\small

\textbf{Input:} $y=(y_1,y_2, ... , y_k) \in \R^{kn}$, $t \geq 1$, $\delta_1,\delta_2 \geq 0$, $1\leq i \leq k$. \\
%
\textbf{Function} $SubR_1(y_1,y_2,...,y_k,\delta_1,\delta_2, i)$ // no second splitting strategy
\begin{algorithmic}[1]

\STATE Set $j_1<j_2<...<j_{k-1}$, where $\{ j_1,j_2,...,j_{k-1} \} = \{1,2,...,k\} \backslash \{i\}$.
 \FOR{each $(t_{j_1},...,t_{j_{k-1}}) \in \mathcal{T}_{j_1}^{\delta_1} \times \mathcal{T}_{j_2}^{\delta_1} ... \times \mathcal{T}_{j_{k-1}}^{\delta_1}$ }
\STATE Compute the sets $\mathcal{V}_i = V \cap B_{\delta_2}\big(y_i-(- \sum_{j'}t_{j_{j'}})\big)$
\FOR{$v_i \in \mathcal{V}_i $}
\STATE Add $(t_{j_1},...,t_{j_{i-1}},v_i+(-\sum_{j'}t_{j_{j'}}), t_{j_{i}}, ..., t_{j_{k-1}})$ to the list $\mathcal{T}$.
\ENDFOR
\ENDFOR
\STATE \textbf{Return} $\mathcal{T}$.

\end{algorithmic}
\textbf{Function} $SubR_2(y_1,y_2,...,y_k,\delta_1,\delta_2, i)$ // with the second splitting strategy
\begin{algorithmic}[1]
\FOR{$1 \leq l \neq i \leq k$}
\STATE Set $j_1<j_2<...<j_{k-2}$, where $\{ j_1,j_2,...,j_{k-2} \} =  \{1,2,...,k\} \backslash \{i,l\}$.
 \FOR{each $(t_{l},t_{j_1},...,t_{j_{k-2}}) \in \mathcal{T}_l^{\delta_1} \times \mathcal{T}_{j_1}^{\delta_1/2} \times \mathcal{T}_{j_2}^{\delta_1/2} ... \times \mathcal{T}_{j_{k-2}}^{\delta_1/2}$ }
\STATE Compute the sets $\mathcal{V}^{\delta_2}_i = V \cap B_{\delta_2}\big(y_i-(- t_l-\sum_{j'}t_{j_{j'}})\big)$
\FOR{$v_i \in \mathcal{V}^{\delta_2}_i $}
\STATE Add $(t_{j_1},...,t_l,...t_{j_{i-1}},v_i+(- t_l-\sum_{j'}t_{j_{j'}}), t_{j_{i}}, ..., t_{j_{k-1}})$ to the list $\mathcal{T}$.
\ENDFOR
\ENDFOR
\ENDFOR
\STATE \textbf{Return} $\mathcal{T}$.
\vspace{1mm}
\end{algorithmic}
\end{algorithm}

\subsection{Decoding paradigm for $\Gamma(V,\alpha,\beta,k)$}

The proposed decoding algorithm simply uses representation~$\eqref{equ_ning_coset_m}$ of the $k$-ing construction:
$\Gamma(V,\alpha,\beta,k)$ is decoded via $|\alpha|$ use of Algorithm~\ref{main_alg_pari}, as described in Algorithm~\ref{main_alg_turing}.
The complexity of Algorithm~\ref{main_alg_turing} is
\begin{align}
\label{equ_complex_alg2}
\mathfrak{C}_{A.\ref{main_alg_turing}}=|\alpha| (k\mathfrak{C}^T_{dec} + k \mathfrak{C}^V_{dec}). 
\end{align}

\begin{algorithm}
\caption{Decoder for $\Gamma(V,\alpha,\beta,k)$}
\label{main_alg_turing} 
\textbf{Input:} $y=(y_1,y_2, ... , y_k) \in \mathbb{R}^{kn}$.   
\begin{algorithmic}[1]
\FOR{$m \in \alpha$}
\STATE $y' \leftarrow y-m^k$
\STATE Use Algorithm~\ref{main_alg_pari} with $y'$ as input.
\ENDFOR
\STATE \textbf{Return} the closest element of $\mathcal{T}$ to $y$.
\end{algorithmic}
\end{algorithm}

\newpage
\section{Decoders for the parity lattices}
\label{sec_parity_lat}





\begin{algorithm}
\caption{Recursive list dec. for $L_{kn}=\Gamma(\theta L_{n},\beta,k)_{\mathcal{P}}\in \R^{kn}$, $n=c \cdot k^{t-1}$} 
\label{main_alg_pari_split_recu_noSp}
\textbf{Function} $ListRecL(y,t,\delta)$ \\
\textbf{Input:} $y=(y_1,y_2, ... , y_k) \in \R^{kn}$, $0 \leq t$, $ 0 \leq \delta$.  
\begin{algorithmic}[1]
\IF{$t=0$}
\STATE $\mathcal{T} \leftarrow$ The set $\Lambda_c \cap B_{\delta}(y)$.
\ELSE
\STATE  $\mathcal{T}_1 \leftarrow ListRecL(y_1,t-1,\delta)$, $\mathcal{T}_2 \leftarrow ListRecL(y_2,t-1,\delta)$,..., 
$\mathcal{T}_k \leftarrow ListRecL(y_k,t-1,\delta)$.
 \FOR{$1 \leq i \leq k$} 
 \STATE Set $j_1<j_2<...<j_{k-1}$, where $\{ j_1,j_2,...,j_{k-1} \} =  \{1,2,...,k\} \backslash \{i\}$.
 \FOR{each $(t_{j_1},...,t_{j_{k-1}}) \in \mathcal{T}_{j_1} \times \mathcal{T}_{j_2} ... \times \mathcal{T}_{j_{k-1}}$ }
\STATE $\mathcal{V}_i \leftarrow ListRecL( [y_i-(- \sum_{j'}t_{j_{j'}})] \cdot R(c \cdot k^t,\theta)^T,t-1,\delta) \cdot R(c \cdot k^t,\theta)$.
\FOR{$v_i \in \mathcal{V}_i $}
\STATE Add 
$(t_{j_1},...,t_{j_{i-1}}, v_i+(-\sum_{j'}t_{j_{j'}}), t_{j_{i}}, ..., t_{j_{k-1}})$
in the list $\mathcal{T}$.
\ENDFOR
\ENDFOR
\ENDFOR
\normalsize
\STATE \textbf{Remove all elements in $ \mathcal{T}$ at a relative distance $> \delta$ from $y$}.
\STATE Sort the remaining elements in $\mathcal{T}$ in a lexicographic order and remove all duplicates.
\ENDIF
\STATE \textbf{Return} $\mathcal{T}$.
\end{algorithmic}
\end{algorithm}

\subsection{Recursive decoding}
\label{sec_recu_dec}
The decoding paradigms presented in the previous section can be adapted to decode lattices recursively built from the single parity-check construction.
As an example, we adapt Algorithm~\ref{main_alg_pari_no_sp} in the recursive Algorithm~\ref{main_alg_pari_split_recu_noSp} 
to decode the parity lattices $L_{kn}=\Gamma(\theta L_{n}, \beta,k)$. 
Hence, we have $T=L_n$ and $V=\theta L_n$.
Since $l(L_n, \delta) = l(\theta L_n, \delta)$, we set $l(n,\delta)=L(n,r)=l(L_n,\delta)$ to simplify the notations. Moreover,
we also write $\mathfrak{C}(\delta)$ for $\mathfrak{C}(\frac{n}{k},\delta)$ and $l(\delta)$ for $l(\frac{n}{k},\delta)$.


In Algorithm~\ref{main_alg_pari_split_recu_noSp} the ``removing"  steps   (Steps~14 and 15  
) are added to ensure 
that a list with no more than $l(n,\delta)$ elements 
is returned by each recursive call. 
This enables to control the complexity of the algorithm (see e.g. Section~\ref{sec_list_dec_BW}).
However, we shall see that the step in bold is not always necessary for the Gaussian channel.

Note that this  algorithm with $\delta = 1/4$ yields a recursive BDD whose complexity is provided by the next theorem.
%
\begin{theorem}
\label{theo_rec_BDD}
Let $n= c \cdot k^t$ and $y \in \R^n$. 
If $d(y,L_{n}) < \rho^2(L_{n})$, then
Algorithm~\ref{main_alg_pari_split_recu_noSp} with $\delta = 1/4$ outputs the closest lattice point to $y$ in time
\begin{align}
\label{equ_complex_big_algo_rec}
\small
\begin{split}
&\mathfrak{C}_{A.\ref{main_alg_pari_split_recu_noSp}}(n,\frac{1}{4}) = O(n^{1+\frac{1}{log_2 k}}).
\end{split}
\end{align}
\end{theorem}

\begin{proof}
\small
\begin{align*}
\mathfrak{C}(n,\frac{1}{4}) = &2k\mathfrak{C}(\frac{n}{k},\frac{1}{4})+ O(n) = O(n)\sum_{i=0}^{\log_k n} \left( \frac{2k}{k} \right)^i =O(n^{1+\frac{1}{log_2 k}}).
\end{align*}
\normalsize
\end{proof}

\subsection{Decoding performance on the Gaussian channel}
\label{sec_dec_gaus}

\begin{lemma}
Let $x \in \Lambda \subset \R^n$ and let $y \in \R^n$ be the point to decode.
Let $\mathcal{T}$ denote the list outputted by a list-decoding algorithm.
The point error probability under list decoding is bounded from above by: 
\small
\begin{align}
\label{eq_error_deco}
\begin{split}
P_e(dec) \leq  &P_e(opt) + P(x \notin \mathcal{T}). \\
\end{split}
\end{align}
\normalsize
\end{lemma}
\begin{proof}
\begin{align}
\small
\begin{split}
P_e(dec) =& P(y \notin \mathcal{V}(x)) + P(x \notin \mathcal{T} \cap y \in \mathcal{V}(x)) \leq  P(y \notin \mathcal{V}(x)) + P(x \notin \mathcal{T}).
 \end{split}
 \normalsize
\end{align}
\end{proof}

In the sequel, we derive formulas to estimate the term $P(x \notin \mathcal{T})$.

\subsubsection{Choosing the decoding radius for regular list decoding on the Gaussian channel}

Consider the Gaussian channel where $y=x+w$, with $y\in \R^{kn}$, $x \in L_{kn}$, and $w \in \R^{kn}$ with i.i.d $\mathcal{N}(0,\sigma^2)$ components.
With a regular list decoder $\mathcal{T} = \Lambda \cap  B_{\delta}(y)$ and
\small
\begin{align}
\begin{split}
P(x \notin \mathcal{T}) = P(||w||^2> r).
\end{split}
\end{align}
\normalsize
Since $||w||^2$ is a Chi-square random variable with $n$ degrees of freedom,
$P(||w||^2> r)=F(n,r,\sigma^2)$, where, for $n$ even :
\small
\begin{align}
\label{equ_chi_squ}
F(n,r,\sigma^2) =  e^{-\frac{r}{2 \sigma^2}} \sum_{k=0}^{n/2 - 1} \frac{1}{k!} \left( \frac{r}{2 \sigma^2} \right)^k.
\end{align}
\normalsize
\begin{lemma}
Consider Algorithm~\ref{main_alg_pari_split_recu_noSp} with the following input parameters.
 The point $y=x+w$, where $y\in \R^{kn}$, $x \in L_{kn}$, and $w \in \R^{kn}$ with i.i.d $\mathcal{N}(0,\sigma^2)$ components. Moreover, $t\geq 0$ and $\delta=r/d(L_{kn})$.
We have  
\begin{align}
P(x \notin \mathcal{T}) = F(kn,r,\sigma^2).
\end{align}
\end{lemma}
Based on \eqref{eq_error_deco}, quasi-optimal performance with regular list decoding is obtained
by choosing a decoding radius $r=E[||w||^2] (1 + \epsilon) =n\sigma^2  (1 + \epsilon)$ such that $F(n,r,\sigma^2)  < \eta \cdot~P_e(opt,\sigma^2)$ (in practice $\eta= 1/2$ is good enough).
Moreover, it is easy to show that $\epsilon \rightarrow 0$ when $n \rightarrow + \infty$.
We denote by $\delta^*$ the relative decoding radius corresponding to this specific $r$:
\small
\begin{align}
 \delta^* = \frac{n \sigma^2 (1 + \epsilon)}{d(\Lambda)}.
\end{align}
\normalsize
Of course, the greater $\delta^*$, the greater the list-decoding complexity.

\subsubsection{A modified list-decoding algorithm}

Notice that due to the ``removing step" 
(Steps 14, in bold, of Algorithms~\ref{main_alg_pari_split_recu_noSp}), if a point found at the last recursive step is at a distance greater than $r$ from $y$, even if it is the unique point found, it is not kept and an empty list is returned: The decoding radius is $r$  in $V$ and $r/2$ in $T$, but only the points at a distance less than $r$ from $y$ are kept. 


To avoid this situation, we remove Step~14 in Algorithm~\ref{main_alg_pari_split_recu_noSp}.
We will see in the rest of the paper that this enables to choose smaller decoding radii for QMLD than with regular list decoding
and reduce the complexity despite the absence of the removing step.
In terms of error probability, decoding in a sphere is the best choice given a finite decoding volume around the received point $y$. 
However, there may be larger non-spherical volumes that achieve satisfactory performance but that are less complex to explore. This is the main idea behind this modified list-decoding algorithm. 
This subsection concentrates on the analysis of the error probability of the modified algorithm. 

\begin{theorem}
\label{theo_error_proba}
Consider Algorithm~\ref{main_alg_pari_split_recu_noSp} without Step~14 with the following input parameters.
The point $y$ is obtained on a Gaussian channel with VNR $\Delta=\vol(L_{kn})^{2/kn}/2 \pi e \sigma^2$ as $y=x+w$, where $y\in \R^{kn}$, $x \in L_{kn}$, and $w \in \R^{kn}$ with i.i.d $\mathcal{N}(0,\sigma^2)$ components. Moreover, $t\geq 0$ and $\delta$ is the relative decoding radius.
We have
\small 
\begin{align}
P(x \notin \mathcal{T}) \leq U_{kn}(\delta,\Delta), 
\end{align}
\normalsize
where
\small 
\begin{align}
\label{eq_with_delta}
\begin{split}
U_n(\delta,\Delta) = &\min \Big\{ \binom{k}{2} U_{\frac{n}{k}}(\delta,\frac{\Delta}{2^{\frac{1}{k}}})^2 +   k  U_{\frac{n}{k}}(\delta, 2^{\frac{k-1}{k}}\Delta)(1- U_{\frac{n}{k}}(\delta , \frac{\Delta}{2^{\frac{1}{k}}}))^{k-1}, \ 1 \Big\}.
 \end{split}
\end{align}
\normalsize
The initial condition $U_c(\delta,\Delta)$ corresponds to the decoding performance in $L_c$: $U_c(\delta,\Delta)= P(x \notin \mathcal{T}_c)$, where $\mathcal{T}_c$ denotes the list of candidates obtained when list decoding in $L_c$.
\end{theorem}
The proof is provided in Appendix~\ref{app_proof_error_proba}.
For instance, a regular list decoder for $L_c$ with relative decoding radius $\delta$ is used in Algorithm~\ref{main_alg_pari_split_recu_noSp}. Consequently, the initial condition is
\small
\begin{align}
\label{equ_ini_cond}
U_c(\delta,\Delta)=F(c,f(\delta),f(\Delta)),
\end{align}
\normalsize
 with $f(\delta) = \delta \cdot d(L_c)$, and $f(\Delta) = \vol(L_c)^\frac{2}{c}/(2 \pi e \Delta)$.

As illustrated by the next example,  \eqref{eq_with_delta} means that the lattices of smaller dimensions are decoded with the same relative radius but with a VNR that is either greater, $2^{\frac{k-1}{k}}\Delta$, or smaller, $\Delta /  2^{\frac{1}{k}} $. 
This result is a consequence of the following properties of the parity lattices: $ \vol(L_n)^\frac{2}{n}= \vol(L_{kn})^{\frac{2}{kn}}/2^{\frac{1}{k}} $ and $d(L_n)=d(L_{kn})/2$ (see Appendix~\ref{app_proof_error_proba} for justifications). 
\begin{figure}
\centering
\includegraphics[width=0.55\columnwidth]{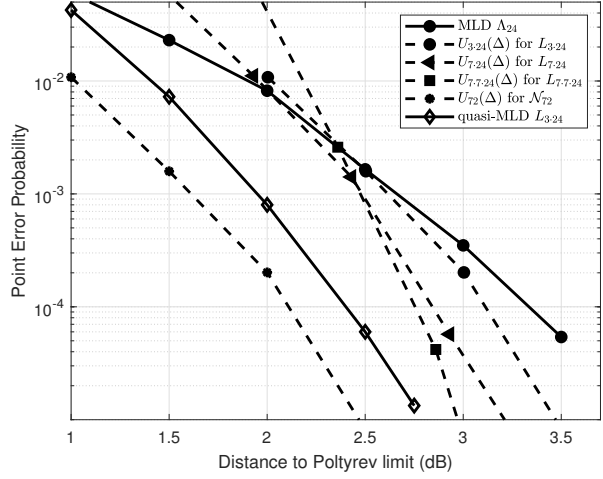}
\caption{Performance curves for Example~\ref{ex_pari_Leech}.}
\label{fig_example_pari}
\end{figure}
\begin{example}
\label{ex_pari_Leech}
Let $\Lambda_{24}$ be the Leech lattice  and let $\beta= [ \Lambda_{24} / \lambda  \Lambda_{24}]$, where $\lambda = \frac{1 + i\sqrt{7}}{2}$  (see Section~\ref{sec_polar} for more details on $\lambda \Lambda_{24}$).
The $k$-parity-Leech lattices are defined as $L_{kn} = \Gamma(\lambda L_n,\beta,k)_{\mathcal{P}}$ 
with initial condition $\Lambda_c = \Lambda_{24}$.
On the Gaussian channel and with Algorithm~\ref{main_alg_pari_split_recu_noSp} without Step~14, the probability that the transmitted lattice point is not in the outputted list is given by \eqref{eq_with_delta}.
We let the initial condition $U_{24}(\Delta)$ be the performance of the optimal decoder for $\Lambda_{24}$
($\delta$ is thus irrelevant in this case). It means that Step~2 of Algorithm~\ref{main_alg_pari_split_recu_noSp} is modified by using a  MLD decoder for $\Lambda_{24}$.

For the lattice $L_{k\cdot 24}$, the value of $U_{k\cdot 24}(\Delta)$ is obtained (using \eqref{eq_with_delta}) by adding the performance curves representing
\begin{itemize}
\item $\binom{k}{2}\cdot (P_e^{\Lambda_{24}}(opt,\Delta))^2$, shifted by $10\log_{10}(2^{\frac{1}{k}})$ dB to the right,
\item and $k \cdot P^{\Lambda_{24}}_e(opt,\Delta)$ shifted by $10\log_{10}(2^{\frac{k-1}{k}})$ dB to the left (assuming that $(1-U_{\frac{n}{k}}(\frac{\Delta}{2^{\frac{1}{k}}}))^{k-1} \approx 1$).
\end{itemize}
The results for $k=3$ and $k=7$ are shown by the dashed line in Figure~\ref{fig_example_pari}.
Since there is only one recursive step, Algorithm~\ref{main_alg_pari_split_recu_noSp} is equivalent to Algorithm~\ref{main_alg_pari_no_sp}. The associated decoding complexity is obtained from \eqref{equ_complex_list_without}
where the term $l(T,\delta)=1$ (note that \eqref{equ_complex_list_without} reduces to \eqref{equ_complex_alg1} in this case). Consequently, we get
\small
\begin{align}
\mathfrak{C}(L_{k \cdot 24}) = 2 k \mathfrak{C}_{MLD}^{\Lambda_{24}}+O( k\cdot 24).
\end{align}
\normalsize
The probability $U_{k \cdot k\cdot 24}(\Delta)$ is obtained in a similar manner from $U_{k\cdot 24}(\Delta)$. For instance, $U_{7 \cdot 7 \cdot 24}(\Delta)$ is also plotted in the figure. The corresponding decoding complexity of  $L_{k \cdot k\cdot 24}$ in this case is obtained from \eqref{equ_complex_list_without} (with $l(T,\delta)=k$, the number of candidates obtained at the previous recursive step):
\small
\begin{align}
\begin{split}
\mathfrak{C}(L_{k\cdot k \cdot 24}) =& k (1+k^{k-1}) \mathfrak{C}(L_{k \cdot 24}) +O( k\cdot k \cdot 24), \\
\approx & 2k^{k+1}\mathfrak{C}_{MLD}^{\Lambda_{24}} + O( k^2  24).
\end{split}
\end{align}
\normalsize
Figure~\ref{fig_example_pari} also depicts the QMLD performance of $L_{3\cdot 24}$ (obtained in Section~\ref{sec_parity72}) for comparison.

We shall see in the next section that the Nebe lattice $\mathscr{N}_{72}$, constructed as $\Gamma(\lambda \Lambda_{24},\alpha,\beta,3)$, has the following properties: $\vol(\mathscr{N}_{72})^{\frac{2}{n=72}}=\vol(T=\Lambda_{24})^{\frac{2}{n/3}}$ and $d(\mathscr{N}_{72})=2d(\Lambda_{24})$. Consequently, \eqref{eq_with_delta} becomes
\small
\begin{align}
\label{eq_with_delta_nebe}
\begin{split}
U_{n=72}(\Delta) = &\min \Big\{ 3 U_{\frac{n}{3}}(\Delta)^2 +  3  U_{\frac{n}{3}}(2 \Delta)(1- U_{\frac{n}{3}}(\delta , \Delta))^{2}, \ 1 \Big\}.
 \end{split}
\end{align}
\normalsize
Taking $U_{\frac{n}{3}}(\Delta)=P_e^{\Lambda_{24}}(opt,\Delta)$, we get a similar curve as $U_{3 \cdot 24}$ for $L_{3\cdot 24}$ but shifted by $10 \cdot \log_{10}(2^{1/3})=1$ dB to the left. The curve $U_{72}(\Delta)$ is shown in Figure~\ref{fig_example_pari}. See Section~\ref{sec_nebe_gaussian_channel} and Figure~\ref{fig_perfNebe_bis} for more details on the quasi-MLD performance of $\mathcal{N}_{72}$.
\end{example}

If the (first) splitting strategy is considered (e.g. in a recursive version of Algorithm~\ref{main_alg_pari_split}), the error probability is slightly greater due to specific cases, such as having simultaneously $\frac{2}{3}\frac{r}{2} < ||w_j||<\frac{r}{2}$ and  $\frac{2}{3} r < ||w_i||<r$, which are not correctly decoded (whereas they were without the splitting strategy).
For the case $k=2$, it is shown in Appendix~\ref{app_proof_error_proba} that with the splitting strategy we get the recursion
\small
\begin{align}
\label{eq_first_spStrat}
\begin{split}
U_n(\delta,\Delta) =  \min \Big\{ U_{\frac{n}{2}}(\delta,\frac{\Delta}{\sqrt{2}})^2 + &2\Big[\big( U_{\frac{n}{2}}(\frac{2}{3}\delta,\frac{\Delta}{\sqrt{2}}) -U_{\frac{n}{2}}(\delta,\frac{\Delta}{\sqrt{2}})\big) U_{\frac{n}{2}}(\frac{2}{3} \delta, \sqrt{2} \Delta) \\
&+ (1-U_{\frac{n}{2}}(\frac{2}{3}\delta, \frac{\Delta}{\sqrt{2}}) )U_{\frac{n}{2}}(\delta,  \sqrt{2} \Delta)  \Big], \ 1 \Big \}.
 \end{split}
\end{align}
\normalsize





\subsection{The 3-parity-Leech lattice in dimension 72}
\label{sec_parity72}

Consider the $3$-parity-Leech lattice $L_{3 \cdot 24} = \Gamma(\lambda \Lambda_{24},[\Lambda_{24}/\lambda \Lambda_{24}],3)_{\mathcal{P}}$ presented in Example~\ref{ex_pari_Leech}.
$L_{3 \cdot 24}$ has the same minimum distance as $\mathscr{N}_{72}$ and a volume $\vol(L_{3 \cdot 24})=\vol(\mathscr{N}_{72})\times|\alpha|$ (using~\eqref{equ_vol}).
Its fundamental coding gain is:
\small
\begin{align}
\gamma(L_{3 \cdot 24}) = \gamma(\Gamma(2S,T_{2\theta},T,3)) \times \frac{1}{(2^{12})^{\frac{2}{72}}} \approx 6.35.
\end{align}
\normalsize
\begin{lemma}
\label{lem_kiss_numb}
The kissing number of $L_{3 \cdot 24}$ is 28,894,320.
\end{lemma}
The proof is provided in Appendix~\ref{app_kissing}. 
The kissing number is about $2^{7.75}$ smaller than the kissing number of $\mathscr{N}_{72}$ (which is $6,218,175,600$).
As a result, one can state the following regarding the relative performance of these two lattices on the Gaussian channel:
1 dB is lost by the parity-Leech lattice due to a smaller $\gamma$, but using the rule of thumb that 0.1 dB is lost each time the kissing number is doubled \cite{Forney1998}, there is also an improvement of 0.8 dB.
Overall, we expect the performance of these two lattices to be only 0.2 dB apart but where the decoding complexity of the 3-parity-Leech lattice is significantly reduced compared to $\mathscr{N}_{72}$ (see Section~\ref{sec_deco_turyn}). The QMLD performance is shown in Figure~\ref{fig_perfNebe_bis} and it is indeed at 0.2 dB from the one of $\mathscr{N}_{72}$.

Consider Algorithm~\ref{main_alg_pari_no_sp} for decoding.
\eqref{eq_with_delta} yields

 \small 
\begin{align}
\begin{split}
U_{3\cdot 24}(\delta,\Delta) = &3 U_{24}(\delta,\frac{\Delta}{2^\frac{1}{3}})^2 + 3  U_{24}(\delta, 2^{\frac{2}{3}} \Delta)(1- U_{24}(\delta , \frac{\Delta}{2^\frac{1}{3}}))^{2}.
 \end{split}
\end{align}
\normalsize
A MLD decoder for $\Lambda_{24}$ as subroutine is not powerful enough to get QMLD performance (see the curve for $U_{3 \cdot 24}$ on Figure~\ref{fig_example_pari}). We can for instance consider a sphere decoder computing $\Lambda_{24} \cap B_{\delta \cdot d(\Lambda_{24})}(y)$.  Then $U_{24}(\delta,\Delta)=F(24,\delta \cdot d(\Lambda_{24}),\sigma^2)$ and the relative decoding radius $\delta^*$ should be chosen such that $3 \cdot F(24,\delta^* \cdot d(\Lambda_{24}),\sigma^2)^2 \approx 1/2 \cdot  P_e^{L_{3\cdot 24}}(opt,\sigma^2)$. We find $\delta^* \approx 25/64$. With Theorem~\ref{lemma_johnson_ori} we get that $l(\Lambda_{24},\delta^*)=4$.
The (worst-case) complexity of Algorithm~\ref{main_alg_pari_no_sp}, given by Lemma~\ref{theo_worst}, becomes

\small
\begin{align}
\begin{split}
\mathfrak{C}^{L_{3\cdot 24}}_{QMLD} =& 3 \mathfrak{C}_{\Lambda_{24}\cap B_{\delta^* \cdot d(\Lambda_{24})}(y) }+  3 l(\Lambda_{24},\delta^*)^2 \mathfrak{C}_{\Lambda_{24}\cap B_{\delta^* \cdot d(\Lambda_{24})}(y)}, \\
=& 51 \cdot \mathfrak{C}_{\Lambda_{24}\cap B_{\delta^* \cdot d(\Lambda_{24})}(y)}.
\end{split}
\end{align}
\normalsize

\subsection{Parity lattices with $k=2$: Barnes-Wall lattices}
\label{sec_pari_k2}

\subsubsection{Existing algorithms}

Several algorithms have been proposed to decode $BW$ lattices. \cite{Forney1988} uses the trellis representation of the two-level squaring construction to introduce an efficient MLD algorithm for the low dimension instances of $BW_n$. 
Nevertheless, the complexity of this algorithm is intractable
for $n>32$: The number edges in the trellis is  $2\cdot 2^{2n/8}+2\cdot2^{3n/8}$, 
e.g. decoding in $BW_{128}$ involves $2\cdot2^{48}+2\cdot 2^{32}$ decoders of $BW_{32}$. Forney states in \cite{Forney1988} : 
``The first four numbers in this sequence\footnote{Forney refers to the number of states per section of the trellis, which is $2^{2n/8}$.}, 
i.e., 2, 4, 16, and 256, are well behaved, but then a combinatorial explosion
occurs: 65 536 states for $BW_{64}$, which achieves a coding gain
of 7.5 dB, and more than four billion states for $BW_{128}$,
which achieves a coding gain of 9 dB. This explosion
might have been expected from capacity and $R_{0}$ (cut-off rate)
considerations''.

Later, \cite{Micciancio2008}~proposed the first BDDs running in polynomial time; 
a parallel version of complexity $O(n^2)$ and a sequential one of complexity $O(n\log^2 n)$.
The parallel decoder was generalized in~\cite{Grigorescu2017} to work beyond the packing radius, still in polynomial time. It is discussed later in the paper. 
The sequential decoder uses the $BW$ multilevel construction to perform multistage decoding:
Each of the $\approx \log n$ levels is decoded with a Reed-Muller decoder of complexity $n\log n$.
This decoder was also further studied, in \cite{Harsham2013}, to design practical schemes for communications over the AWGN channel.
However, the performance of this sequential decoder is far from MLD.
A simple information-theoretic argument explains 
why multistage decoding\footnote{Where only one candidate is decoded at each level.} of $BW$ lattices cannot be efficient: 
The rates of some component Reed-Muller codes
exceed the channel capacities of the corresponding levels \cite{Forney2000}\cite{Yan2013}.

As a result, no $BW$ decoder, being both practical and quasi-optimal on the Gaussian channel, 
have been designed and executed for dimensions greater than $32$.

\subsubsection{A new BDD}
\label{sec_BW_BDD}
Algorithm~\ref{main_alg_pari_split_recu_noSp} with $k=2$ and $\delta=1/4$ yields a new BDD for the Barnes-Wall lattices.
%
%
%
%
We apply Theorem~\ref{theo_rec_BDD} for the case $k=2$. 
\begin{corollary}
\label{theo_complex_1}
Let $n=2^{t+1}$ and $y\in \R^n$. If $d(y,BW_n)<\rho^2(BW_n)$, 
then Algorithm~\ref{main_alg_pari_split_recu_noSp} with $k=2$ and $\delta=1/4$ outputs the closest lattice point to $y$ in time $O(n^2)$.
\end{corollary}

\subsubsection{List decoding  $BW$ lattices}
\label{sec_list_dec_BW}

The recursive version of Algorithm~\ref{main_alg_pari_split}, with $k=2$, can be used to list decode the $BW$ lattices. For regular list decoding the removing and sorting steps, as in Algorithm~\ref{main_alg_pari_split_recu_noSp}, should also be added. Let us name it Algorithm~\ref{main_alg_pari_split}'. 

We investigate the complexity of the algorithm of \cite{Grigorescu2017} and Algorithm~\ref{main_alg_pari_split}'.
As mentioned at the beginning of the section, \cite{Grigorescu2017} adapts the parallel BDD of \cite{Micciancio2008}, which uses 
the automorphism group of $BW_n$, to output a list of all lattice points lying at a distance $r =d(BW_n)(1 - \epsilon)$, $0<\epsilon \leq 1$,  
from any $y\in \R^n$ in time 
\begin{align}
\label{complex_algo_grigo}
O(n^2) \cdot L(n,r^2)^2. 
\end{align}
A critical aspect regarding the complexity of this decoder is therefore the list size.
Theorem~\ref{lemma_johnson_ori} provides bounds on the list size when $r\leq d(BW_n)/2$. 
The following lemma, addressing $r> d(BW_n)/2$, is proved in \cite{Grigorescu2017}.
%
\begin{lemma}[Results from \cite{Grigorescu2017}]
\label{lemma_johnson}
The list size of $BW_n$ lattices is bounded as:  
\begin{itemize}
\item $L(n,r) = O(n^{\log_2 4 \lfloor \frac{3}{4\epsilon} \rfloor })$
 if $r \leq d(BW_n)(\frac{3}{4} - \epsilon)$, $0~<~\epsilon~<~\frac{1}{4}$.
\item $L(n,r) = O(n^{2 \log_2 24})$ if $r = \frac{3}{4} d(BW_n)$.
\item And
$L(n,r) = O( n^{8 \log_2 \frac{1}{\epsilon}})$,
if $ r \leq d(BW_n)(1 - \epsilon)$, $0~<~\epsilon~<~\frac{1}{4}$.
\end{itemize}
\end{lemma}

Lemma~\ref{lemma_johnson} shows that the list size of $BW$ lattices is of the form $n^{O( \log \frac{1}{\epsilon})}$ and thus polynomial in the lattice dimension for any radius bounded away from the minimum distance. 
Combining the lemma with \eqref{complex_algo_grigo}, the list decoder complexity becomes $n^{O(\log \frac{1}{\epsilon})}$ for any $r<d(BW_n)(1-\epsilon)$, $\epsilon > 0$. 
This result is of theoretical interest: 
It proves that there exists a polynomial time decoding algorithm (in the dimension)
for any radius bounded away from the minimum distance. 
However, the quadratic dependence is a drawback:
As already explained, 
finding an algorithm with quasi-linear dependence in the list size is stated as an open problem in \cite{Grigorescu2017}.
%


In the following, we demonstrate that if we use Algorithm~\ref{main_alg_pari_split}', rather 
than the automorphism group of $BW_n$ for list decoding,
we get complexity linear in the list size. 
This enables to both improve the regular list-decoding complexity 
and get a practical quasi-optimal decoding algorithm on the Gaussian channel up to $n=128$. 

We compute below the complexity of our algorithm for $\delta~<~9/16$. The complexity analysis for larger $\delta$ (which is the proof of Theorem~\ref{theo_main_complex_sp}) is provided in Appendix~\ref{app_proof_list}. 

If $\delta < \frac{3}{8}$ then $\frac{2}{3}\delta < \frac{1}{4}$ and we have $l(\delta)=O(1)$, $l(\frac{2}{3}\delta)=1$. Moreover, $\mathfrak{C}(\frac{2}{3} \delta) \leq \mathfrak{C}(\frac{1}{4}) = O(n^{2})$ (Theorem~\ref{theo_complex_1}). 
The complexity becomes
\begin{align}
\label{equ_exple1}
\small
\begin{split}
\mathfrak{C}(n,\delta) = & 4 \mathfrak{C}(\delta) + l(\delta)O(n^2)= l(\delta) O(n^2 \log n) = \widetilde{O}(n^2).
\end{split}
\end{align}
 If $\frac{3}{8} \leq \delta <\frac{9}{16}$ and $\mathfrak{C}(\frac{2}{3} \delta) \leq \mathfrak{C}(\frac{3}{8}) = O(n^2\log n)$. We get
\begin{align}
\label{equ_exple2}
\small
\begin{split}
\mathfrak{C}(n,\delta)  =&  l(\delta) O(n^{2} \log n) \sum_{i=0}^{\log_2 n } \left( \frac{2 l(\frac{2}{3} \delta) + 2}{4} \right)^i = l(\delta) O(n^{1+\log_2[1+l(\frac{2}{3} \delta)]} \log n) = l(\delta) \tilde{O}(n^{1+\log_2[1+l(\frac{2}{3} \delta)]}), \\
\end{split}
\end{align}
which is $\tilde{O}(n^{1+\log_2 3})$ if $\delta<1/2$.

Note that for these cases ($\delta < 1/2$) the decoder of \cite{Grigorescu2017} is more efficient: Indeed, Theorem~\ref{lemma_johnson_ori} shows that when $\delta<1/2$ then $l(n,\delta)=O(1)$
and the decoding complexity, given by \eqref{complex_algo_grigo}, is $O(n^2)$. 
Nevertheless, the following theorem (proved in Appendix~\ref{app_proof_list}) shows that our decoder is better for larger values of $\delta$ and,
as we shall see in the next subsection, is useful even when $\delta<1/2$ for quasi-optimal decoding on the Gaussian channel.

\begin{theorem}
\label{theo_main_complex_sp}
Let $n= 2^{t+1}$, $y \in \R^n$. 
The set $BW_n \cap B_{\delta}(y)$ can be computed in worst-case time:
\begin{itemize}
\item $O(n^{2})$ if $\delta < \frac{1}{2}$ (algorithm of \cite{Grigorescu2017}).
\item $l(\delta) O(n^{2 + \log_2 [\frac{l(\frac{2}{3} \delta)+1}{2}]} )  \approx O(n^{1+\log_2 4\lfloor \frac{3}{4 \epsilon} \rfloor^2]})$ if $\delta=\frac{3}{4}-\epsilon$, $0<\epsilon$. 
\item  $O(n^{1+ \log_2 432})$ if $\delta = \frac{3}{4}$.
\item $l(\delta)O(n^2) = O(n^{8 \log_2 \frac{1}{\epsilon} +2})$ if $\delta=1-\epsilon$, $0<\epsilon<\frac{1}{4}$. 
\end{itemize}
\end{theorem}

\subsubsection{Decoding on the Gaussian channel}
\label{sec_gauss_BW}

We apply the analysis presented in Section~\ref{sec_dec_gaus} to the case $k=2$ to establish the smallest list-decoding radius $\delta$ required for quasi-optimal decoding.
The first element needed is the MLD performance $P_e^{BW_n}(opt,n,\sigma^2)$ of $BW_n$.
As mentioned earlier it is not known for $n>32$.
Nevertheless, $P_e^\Lambda(opt,n,\sigma^2)$ can be lower-bounded  for $any$ lattice $\Lambda$ in $n$ dimensions using the sphere lower bound \cite{Tarokh1999} (see also \cite{Forney2000} or \cite{Ingber2013}). Table~\ref{table_dist_polty} provides the sphere lower bound on the best performance achievable for $P_e^{\Lambda}(opt,n,\sigma^2)=10^{-5}$. 
With \eqref{equ_chi_squ}, we can compute the smallest $\delta^*$, for the corresponding values of $\sigma^2$, such that $P(x \notin \mathcal{T}) = P(||w||^2> r)\lessapprox 10^{-5}$ with regular list decoding. Using $\delta^*$ yields quasi-optimal decoding performance, regardless of the MLD performance of $BW_n$. 
The values of $\delta^*$ as a function of $n$ are presented in Figure~\ref{fig_rad_mu}. 
The corresponding (worst-case) decoding complexity is obtained with Theorem~\ref{theo_main_complex_sp}. 
It is super-quadratic for all $n\geq16$. 

\begin{table}
\begin{center}
\begin{tabular}{|c|c|c|c|c|c|}
 \hline
Dimension $n$  & 16 & 32 &  64 & 128 & 256  \\
  \hline
   &  &  &   & & \\
Dist. to Polt. (dB) sphere bound & 4.05 & 3.2 & 2.5 & 1.9 & 1.4 \\
  \hline
   &  &  &   & & \\
Dist. to Polt. (dB) MLD & 4.5 & 3.7 & 3.1 & 2.3 & ? \\
\hline

\end{tabular}
\end{center}
\caption{Sphere lower bound on the best performance achievable by any lattice $\Lambda$ for  $P_e^{\Lambda}(opt,n,\sigma^2)=10^{-5}$ and MLD performance of $BW_n$ for $P_e(opt,n,\sigma^2)=10^{-5}$.}
\label{table_dist_polty}
\end{table}

Running the simulations (with $\delta^*$ found at the sphere bound) enables to estimate the MLD performance of $BW_n$ lattices. The results are presented\footnote{These estimations were not performed with the regular list decoder, but with the algorithm presented in the rest of the section.} in Table~\ref{table_dist_polty}, and are at $\approx 0.5$ dB of the sphere bound\footnote{We have not yet investigated the case $n=256$.}.  
In Figure~\ref{fig_rad_mu} the corresponding values of $\delta^*$ (still with regular list decoding) are depicted by the diamonds.
For $n=128$, we have $\delta^*>3/4$. Even though the complexity of the regular list decoder is linear in the list size, it remains  $O(n^{8 \log_2 \frac{1}{\epsilon}+2})$ for $\delta = 1-\epsilon$, $0<\epsilon<1/4$. This is much too high to run on a computer in a reasonable time (linear in the list size doesn't mean low complexity in this case).  Consequently, we consider a slightly different strategy to get a practical algorithm for the Gaussian channel.

\begin{figure}
\centering
\includegraphics[width=0.45\columnwidth]{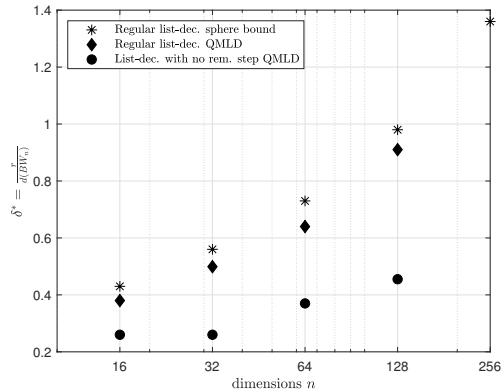}
\caption{Values of the list-decoding relative radius $\delta^*$, for $BW_n$, such that $P(x \notin \mathcal{T}) \approx 10^{-5}$.}
\label{fig_rad_mu}
\end{figure}


As explained in Section~\ref{sec_dec_gaus}, the error probability  of the list-decoding algorithm without the removing step can be estimated with Equation~\eqref{eq_with_delta} without the splitting strategy or \eqref{eq_first_spStrat} with the first splitting strategy.  
Hence, we can also compute the smallest $\delta^*$ such that with this algorithm $P(x \notin \mathcal{T}^{\delta}) \lessapprox 10^{-5}$ (at the MLD performance). 
However, the decoding complexity should be updated to take into account
the fact that there is no removing step in the algorithm even when $\delta<1/2$. 

We consider Algorithm~\ref{main_alg_pari_split}' without the removing step. 
To mitigate the complexity and simplify the analysis, 
whenever $\delta \leq 1/4$ we shall use the BDD presented in Section~\ref{sec_BW_BDD} (i.e. when $\delta \leq 1/4$, we fix it to $1/4$).
Hence, in \eqref{eq_first_spStrat}, $U_n(\delta\leq \frac{1}{4},\Delta)=P_e(BDD,n,\Delta)$. 
The error probability $P_e(BDD,n,\Delta)$ is shown in Figure~\ref{fig_BWperf}.
In the literature,  the performance of BDDs is often estimated via the ``effective error coefficient" \cite{Forney1996} \cite{Salomon2006}. Nevertheless, it it not always accurate, especially in high dimensions. We therefore rely on the Monte Carlo simulations presented in the figure for $P_e(BDD,n,\Delta)$.
The estimated $\delta^*$ with this decoder, shown in Figure~\ref{fig_rad_mu}, are significantly smaller than the ones obtained with the regular list decoder. In particular, $\delta^*<3/8$ for $n\leq 64$ and $\delta^* < 1/2$ for $n=128$.

\begin{figure}[t]
\centering
\includegraphics[width=0.6\columnwidth]{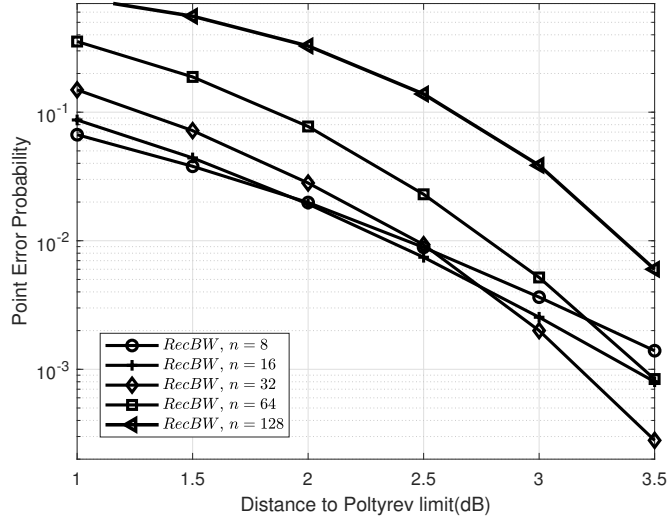}
\caption{Performance of the recursive BDD for the Barnes-Wall lattices on the Gaussian channel.}
\label{fig_BWperf}
\end{figure}

We now study the complexity of this latter algorithm.
We shall use the notation $l'(n,\delta,y)$ to denote the number of elements returned by the decoding algorithm (without a removing step). 
If $\delta \leq 3/8$, we get 
\small
\begin{align}
\begin{split}
l'(n,\delta)&\leq 2 \big[l'(\frac{1}{4})  l'(\delta ) + l'(\delta)  l'(\frac{1}{4})] =4 l'(\delta ) = 4^{\log_2 n} \cdot l(\mathbb{Z}_2,\delta) =O(n^2).
\end{split}
\end{align}
\normalsize
However, considering the average complexity, and taking into account the fact that we remove the duplicates at each recursive step, one has 
\small
\begin{align}
\begin{split}
E_y[l'(n,\delta,y)] \leq & 2 \big[l'(\frac{1}{4}) E_y[ l'(\delta )] + E_y[l'(\delta)]  l'(\frac{1}{4}) - l'(\frac{1}{4}) l'(\frac{1}{4}) ] -E_y^c. \\
=&4 E_y[ l'(\delta )] - 2 - E_y^c,
\end{split}
\end{align}
\normalsize
where $E_y^c$ denotes the average number of common elements in the lists returned by the recursive calls.  
We observed experimentally that for $\delta \leq 3/8$, $E_y[l'(n,\delta,y)]$ is close to 1. 
This observation is not taken into account in the next theorem,
 which bounds the average list size and the average complexity.
 It is however in the interpretation following the theorem.

\begin{theorem}
\label{theo_ave}
Let $E_{y}[l'(n,\delta,y)]$ be the average list size of Algorithm~\ref{main_alg_pari_split}' without the removing step. 
Let $\eta$ denote $E_y[l'(n,3/8,y)]$. 
If $3/8 <\delta \leq 9/16$, $E_{y}[l'(n,\delta,y)]$ is bounded from above as
\begin{align}
E_{y}[l'(n,\delta,y)] = O(n^{2 + \log_2 \eta }).
\end{align}
And the average complexity is bounded from above as:
\begin{itemize}
\item  $E_{y}[\mathfrak{C}(n,\delta)] = \eta \widetilde{O}(n^{2})$  if $\delta \leq 3/8$.
\item $E_{y}[\mathfrak{C}(n,\delta)] = E_{y}[l'(\delta,y)] \widetilde{O}(n^{1+ \log_2[1 + \eta ]})$  if $3/8 <\delta \leq 9/16$.
\end{itemize}
\end{theorem}
See Appendix~\ref{app_theo_ave} for the proof. \\
As a result, based on the observation that $\eta$ is close to 1, 
the average complexity is estimated as:
\begin{itemize}
\item  $E_{y}[\mathfrak{C}(n,\delta)] = \widetilde{O}(n^{2})$  if $\delta \leq 3/8$.
\item $E_{y}[\mathfrak{C}(n,\delta)] =  \widetilde{O}(n^{4})$  if $3/8 <\delta \leq 9/16$.
\end{itemize}

Since $\delta^*<3/8$ for $n<64$ and $\delta^* < 1/2$ for $n=128$ for QMLD, we conclude that the decoding complexity is quadratic for $n\leq 64$ and quartic for $n=128$.

For a practical implementation, we can bound the maximum number of points kept at each recursive step: 
I.e. at the end of each recursive call, 
the $\aleph(\delta)$ best candidates are kept at each recursive step.
The size of the list $\aleph(\delta)$, for a given $\delta$, is a parameter to be fine tuned:
For $n=$16, 32, 64, we set $\aleph(\delta)=$5, 10, 20, respectively. 
For  $n=128$, $\aleph(\delta)=1000$ ($<<$ than our bound in $O(n^2)$ on $E_y[l'(\delta,y)]$) and $\aleph(2/3 \delta)=4$ yields quasi-MLD performance. 
Figure~\ref{fig_perf_list} depicts the simulation results for $BW$ lattices up to $n=128$. \\


\begin{figure}[t]
\centering
\includegraphics[width=0.6\columnwidth]{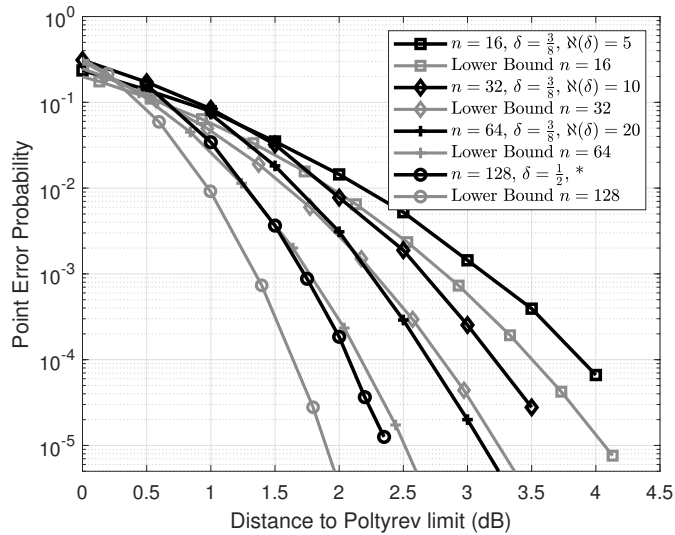}
\vspace{-2mm}
\caption{Simulation results for the $BW$ lattices up to $n~=~128$
and the universal bounds of \cite{Tarokh1999}. $^*$For $n=128$, $\aleph(\delta)=1000$ and $\aleph(2/3 \delta)=4$.}
\label{fig_perf_list}
\end{figure}

\section{Decoders for Leech and Nebe lattices}
\label{sec_deco_turyn}

\subsection{Existing decoding algorithms for $\Lambda_{24}$ and $\mathscr{N}_{72}$}
\label{sec_deco_algo}

\subsubsection{History of the decoders of $\Lambda_{24}$}

$\Lambda_{24}$ appeared under many different forms in the literature (which may be equivalent to Turyn's construction).
Among  others, $\Lambda_{24}$ can be  obtained as (i)~8192 cosets  of $4D_{24}$, (ii)
4096   cosets
of $(\sqrt{2}E_{8})^3$,  (iii)~2~cosets of  the  half-Leech lattice  $H_{24}$,
where  $H_{24}$ is  constructed  by applying  Construction~$B$ on  the
Golay code $C_{24}$,  and (iv) 4 cosets of  the quarter-Leech lattice,
where quarter-Leech lattice  is also built  with Construction~$B$ but applied  on a
subcode of $C_{24}$. Finally, one of the simplest constructions is due
to~\cite{Bonnecaze1995},  where the  Leech  lattice  is obtained  via
Construction A applied on the  quaternary Golay code.

The history of
maximum-likelihood decoding (MLD) algorithms   for  $\Lambda_{24}$  starts
with~\cite{Conway1984}, where  Conway and  Sloane used (i)  to compute
the second moment  of the Voronoi region of  $\Lambda_{24}$. The first
efficient  decoder  was presented  in~\cite{Conway1986}  by the  same
authors using construction  (ii). Two years later,  Forney reduced the
complexity of  the decoder by  exploiting the same  construction (ii),
which he rediscovered in the scope of the ``cubing construction", with
a 256-state  trellis diagram representation~\cite{Forney1988} (see Section~\ref{sec_trellis_dec} for a presentation of trellis).  A year
later, it was further improved in~\cite{Lang1989} and~\cite{Beery1989}
thanks   to   (iii)~combined    with   an   efficient   decoder   of
$C_{24}$. Finally, (iv)~along with the hexacode is used  to build the
fastest ever known MLD decoder by Vardy and Be'ery~\cite{Vardy1993}.

To further reduce the complexity, (suboptimal) BDD
were   also   investigated   based    on   the   same   constructions:
e.g.~\cite{Forney1989}    with    (iii)    and~\cite{Amrani1994}\cite{Vardy1995}\cite{Forney1996} with~(iv).
In these papers, it is shown that these BDD
do not change the  error exponent (i.e. the  effective minimum distance
is not diminished) but increase the ``equivalent error coefficient". The
extra loss is roughly 0.1 dB on the Gaussian channel compared to the optimal performance. \\
As we shall see in the sequel, our decoding  paradigm applied to the Leech lattice
is  more   complex  than  the   state-of-the-art  decoders   of  Vardy
\cite{Vardy1995}\cite{Forney1996} which requires only $\approx 300$
real operations.  But again,
this  latter decoder  is  specific  to the  Leech  lattice whereas  our
decoder is more  universal as it can be used,  among others, to decode
the Nebe lattice and the Barnes-Wall lattices. 

\subsubsection{The decoder of the Nebe lattice in \cite{Meyer2011}}
\label{sec_meyer_deco}
While the decoding of $\Lambda_{24}$ has been extensively studied, the literature on decoders for $\mathscr{N}_{72}$ is not as rich:
Only \cite{Meyer2011} studied this aspect, but the proposed decoder is highly suboptimal. 

First, notice that we can multiply (on the left) the matrix $Pb$ given in~\eqref{eq_complex_pb}
by a unimodular matrix to get the following matrix $Pb'$:
\begin{align}
Pb'=
\left[
\begin{matrix}
1 & 1 & \lambda \\
0 & \psi & \psi \\
0 & 0  & 2
\end{matrix}
\right]
=
\left[
\begin{matrix}
1 & 1 & 0 \\
0 & 0 & 1 \\
\psi & -\lambda & 0
\end{matrix}
\right]
\cdot
Pb.
\end{align}
Similarly to \eqref{eq_complex_pb},  $Pb'\otimes G_{S}^{\mathbb{C}}$, $S \cong \Lambda_{24}$, is a basis for the Nebe lattice which induces the following structure:
\begin{align}
\begin{split}
\mathscr{N}_{72} = & \{ (a,b,c) \in \mathbb{C}^{36}:  a \in S, b-a \in T, c - (b-a) - \lambda a \in 2S\},
\end{split}
\label{equ_succ}
\end{align}
where \eqref{equ_succ} is derived from the columns of $Pb'$.
A successive-cancellation-like  algorithm can thus be  considered: given
$y=(y_1,y_2,y_3)$ in $\mathbb{C}^{36}$, $y_1$  is first decoded in $S$
as  $t_1$,   $y_2-t_1$  is   then  decoded  in   $T$  as   $t_2$,  and
$y_3-t_2-\lambda t_1$ is   decoded   in   $2S$   as   $t_3$.    In
\cite{Meyer2011},  this successive-cancellation  algorithm is  proposed,
with  several  candidates for  $t_1$  which  are obtained  via  sphere
decoding with a given radius  $r$. Among all resulting approximations,
the closest to $y$ is kept.  It is proved in \cite{Meyer2011}, that
the lattice point $\hat{x}$ outputted  by the algorithm using a decoding
radius  $r=R(S)$, the covering radius of $S$, has  an  approximation  factor $||y-\hat{x}||  \leq
\sqrt{7} ||y-x_{opt}||$. Additionally, this  algorithm is guaranteed
to output the closest point $x_{opt} \in \turing$ to $y$
if $d(y,x_{opt}) \leq R(S)$, where $R(S)$ is unfortunately
smaller by a factor $\sqrt{2}$ than the packing radius $\rho(\mathscr{N}_{72})$.

\subsection{New decoders for $\turing$}
\label{BDD_k3}


%
%
%
%
%

We first adapt Algorithm~\ref{main_alg_turing} to $\turing$
by choosing the decoders for $T$ and $V=2S$ as BDDs.
We name it Algorithm~{\ref{main_alg_turing}'}.




\begin{theorem}
\label{th_BDD1_algo3}
Let  $\turing$ and $y$ be respectively a lattice and a point in $\mathbb{R}^{3n}$.
If $d(y,\turing)< \rho^2(\turing)$, then Algorithm~{\ref{main_alg_turing}'} outputs the closest lattice point $x \in \turing$ to $y$ in time
\small
\begin{align}
\label{complex_alg_2}
\mathfrak{C}_{A.\ref{main_alg_turing}'}=6 |\alpha|  \mathfrak{C}_{BDD}^{S}.
\end{align}
\normalsize
\end{theorem}

\begin{proof}
We first show that $x$ is the closest lattice point to $y$. 
Assume that, at Steps 1-2, $m$ corresponds to the coset of the closest lattice point to $y$.
Then, the result follows from Theorem~\ref{theo_worst} since Algorithm~{\ref{main_alg_turing}'} is a special case of Algorithm~\ref{main_alg_pari_no_sp} used $\alpha$ times. \\
Regarding the complexity, we use Equation~\eqref{equ_complex_alg2} with $k=3$ and where $\mathfrak{C}^S_{BDD}=\mathfrak{C}^{2S}_{BDD}$.
\end{proof}
It is insightful to compare Algorithm~{\ref{main_alg_turing}'} to trellis decoding. 
The complexity is reduced from $\approx |\alpha||\beta|^2\mathfrak{C}^S_{CVP}$ to $\approx |\alpha|\mathfrak{C}^S_{BDD}$ (but where trellis decoding is optimal unlike Algorithm~{\ref{main_alg_turing}'}). 

We name Algorithm~{\ref{main_alg_pari_split}''} the list-decoding version of Algorithm~{\ref{main_alg_turing}'}: It consists in repeating $|\alpha|$ times (once for each coset of $\Gamma(2S,\beta,3)_{\mathcal{P}}$) Algorithm~\ref{main_alg_pari_split}, with $k$=3, using the first splitting strategy and the second splitting strategy (i.e. the function $SubR_2$). We use Lemma~\ref{theo_set_split} to get the following theorem.


\begin{theorem}
\label{theo_complex_turing}
Let $\turing$ and $y$ be respectively a lattice and a point in $\mathbb{R}^{3n}$.
Algorithm~{\ref{main_alg_pari_split}''} outputs the set $\turing \cap B_{\delta}(y)$ in worst-case time
\small
\begin{align}
\label{equ_dec_complex_list}
\begin{split}
\mathfrak{C}_{A.\ref{main_alg_pari_split}''}(\delta) = &  |\alpha|  \big[3\mathfrak{C}_{T\cap B_{\delta}(y)}  +  6 l(T,\delta) l(T,\frac{\delta}{2}) \mathfrak{C}_{V\cap B_{\frac{2}{3}\delta}(y)}  + 6 l(T,\frac{2}{3} \delta) l(T,\frac{\delta}{3}) \mathfrak{C}_{V\cap B_{\delta}(y)} \big] .
\end{split}
\end{align}
\normalsize
\end{theorem}

\begin{corollary}
Let $\Lambda_{24}=\turing$ (constructed as in Lemma~\ref{th_turyn}). Algorithm~\ref{main_alg_pari_split}''  with a decoding radius $r=d(T)=d(E_8)$, i.e. $\delta = d(E_8)/d(\Lambda_{24})=1/2$, solves the CVP for $\Lambda_{24}$ with worst-case complexity
\small
\begin{align}
\begin{split}
\mathfrak{C}_{A.\ref{main_alg_pari_split}'}(\delta=\frac{1}{2}) = & |\alpha| \big[ 3\mathfrak{C}_{T\cap B_{\delta}(y)} + 6[ 2 n 2 \mathfrak{C}_{E_8\cap B_{\frac{1}{3}}(y)}  + 3 \mathfrak{C}_{E_8 \cap B_{\frac{1}{2}}(y)} ] \big], \\
 \lessapprox &  2^4 \cdot 6 \cdot 2 \cdot 8 \cdot 2 \cdot \mathfrak{C}_{E_8\cap B_{\frac{1}{2}}(y)} \approx 2^{11} \mathfrak{C}_{E_8\cap B_{\frac{1}{2}}(y)} . 
\end{split}
\end{align}
\normalsize
\end{corollary}
\begin{proof}
If $S \cong E_8$, $d(T) = R^2(\turing)$ (the covering radius).
\end{proof}

To the best of our knowledge, the covering radius of $\mathscr{N}_{72}$ appears nowhere in the literature.
However, Gabriele Nebe showed in a private communication that it is greater than $\sqrt{2} \rho(\mathscr{N}_{72})$. The proof is available in Appendix~\ref{app_cov}. 
As a result, Algorithm~\ref{main_alg_pari_split}'' with $\delta = 1/2$ is not optimal for $\mathscr{N}_{72}$. The algorithm should be used with greater $\delta$ to ensure optimality.

\subsection{Decoding $\Lambda_{24}$ and $\mathcal{N}_{72}$ on the Gaussian channel}
\label{sec_nebe_gaussian_channel}
The analysis is similar to the one performed for $BW$ lattices in Section~\ref{sec_gauss_BW}.
We will therefore be brief on the explanations.

The sphere lower bound for $P_e^{\Lambda}(opt,n,\sigma^2)=10^{-4}$ in dimension 72 yields a distance to Poltyrev limit of 2.1~dB. The MLD performance of $\Lambda_{24}$ for this error probability is 3.3 dB.
Regarding the relative radius to ensure $P(x \notin \mathcal{T}^{\delta}) = P(||w||^2> r)\lessapprox 10^{-4}$ with regular list decoding, we find with Equation~\eqref{equ_chi_squ} $\delta^* \approx 0.57$ for $\mathcal{N}_{72}$ and $\delta^* \approx 0.41$  for $\Lambda_{24}$. 

An important observation (also made at the end of Example~1) when computing the performance of the modified list decoders on the Gaussian channel is the following. Let $T \in \R^n$. For $\Lambda_{24}$ and $\mathscr{N}_{72}$ constructed as $\turing$, we have (see e.g. the proof of Theorem~\ref{th_turyn}) $\vol(\turing)^{\frac{2}{3n}} = \vol(T)^\frac{2}{n}$,
whereas for the parity lattices, we have $\vol(L_{kn})^{\frac{2}{3n}} = 2^{\frac{1}{k}} \vol(T)^\frac{2}{n}$.
This means that the equivalent VNR $\Delta$ is the same when decoding in $\turing$ and in $T$. This will be taken into account in the next formulas to estimate~$\delta^*$.

\noindent \textbf{Decoding $\Lambda_{24}$.}\\
Considering the list-decoding version of Algorithm~{\ref{main_alg_turing}'} (without the splitting strategy), \eqref{eq_with_delta} becomes (see the proof of Theorem~\ref{theo_error_proba})
 \small 
\begin{align}
\label{eq_perf_leech}
\begin{split}
U_{24}(\delta,\Delta) = & \min \{ 3 U_{8}(\delta,\Delta)^2 + 3  U_{8}(\delta, 2 \Delta)(1- U_{8}(\delta , \Delta))^{2},1\}.
 \end{split}
\end{align}
\normalsize
Assume that $\delta^* \leq \frac{1}{4}$ with this algorithm. 
If this holds, $T,V\cong E_8$ can be decoded with the recursive BDD discussed in Section~\ref{sec_pari_k2} (since $E_8 \cong BW_8$). Hence,
$U_{8}(\frac{1}{4},\Delta)=P_e(BDD,\Delta)$ is given by the curve $n=8$ in Figure~\ref{fig_BWperf}.
With \eqref{eq_perf_leech}, for $\Delta=3.3$ dB we find $U_{24}(\delta=\delta^*,\Delta) \leq 10^{-4}$, which confirms that
$\delta^*<1/4$. 
As a result, we can use Algorithm~{\ref{main_alg_turing}'} for quasi-MLD decoding of $\Lambda_{24}$. 
The complexity  of Algorithm~{\ref{main_alg_turing}'} is 
\small
\begin{align}
\label{equ_bloublou}
\begin{split}
\mathfrak{C}_{QMLD}^{\Lambda_{24}}= \mathfrak{C}_{A.{\ref{main_alg_turing}'}}(\Lambda_{24},\delta =\delta^* )  = & 2^4  (3 \mathfrak{C}(E_8) + 3 \mathfrak{C}(RE_8)) = 96 \mathfrak{C}(E_8) .
\end{split}
\end{align}
\normalsize

\noindent \textbf{Decoding $\mathscr{N}_{72}$.}\\
Regarding $\mathscr{N}_{72}$, with the first decoder we have (similar to \eqref{eq_perf_leech})
\small 
\begin{align}
\label{eq_perf_nebe}
\begin{split}
U_{72}(\delta,\Delta) = & \min \{ 3 U_{24}(\delta,\Delta)^2 + 3  U_{24}(\delta, 2 \Delta)(1- U_{24}(\delta , \Delta))^{2},1\}.
\end{split}
\end{align}
\normalsize
Consider a MLD decoder for $\Lambda_{24}$ such that $U_{24}(\delta,\Delta)=P_e^{\Lambda_{24}}(opt,\Delta)$. Then, when $\Delta>1$, $U_{72}(\delta,\Delta) \approx 3 (P_e^{\Lambda_{24}}(opt,\Delta))^2$. The performance of this decoder for $\mathcal{N}_{72}$ is shown by the curve $U_{72}(\Delta)$ on Figure~\ref{fig_example_pari} in Example~\ref{ex_pari_Leech}. Unlike for the parity lattices, the curve for $P_e^{\Lambda_{24}}(opt,\Delta)$ should not be shifted to the right before squaring, as explained in Example~\ref{ex_pari_Leech}. We easily see that this decoder is powerful enough to get quasi-MLD performance for $\mathscr{N}_{72}$. 
 The complexity is then
 \small 
\begin{align}
\label{eq_perf_leech_frfr}
\begin{split}
\mathfrak{C}_{QMLD}^{\mathscr{N}_{72}}  = &2^{12} \cdot [3\mathfrak{C}_{MLD}^{\Lambda_{24}}+ 3\mathfrak{C}_{MLD}^{\Lambda_{24}} ] = 2^{12} \cdot 6 \cdot \mathfrak{C}_{MLD}^{\Lambda_{24}}.
 \end{split}
\end{align}
\normalsize

The curve of quasi-optimal performance of $\mathscr{N}_{72}$ on the Gaussian channel is depicted in Figure~\ref{fig_perfNebe_bis}.  The figure also shows the performance of $L_{3 \cdot 24}$ (discussed in Section~\ref{sec_parity72}). The performance of  $\mathscr{N}_{72}$ is at a distance of 2.6 dB only from Poltyrev limit at around $10^{-5}$ of error per point. 


\begin{figure}
    \centering
	\includegraphics[width=0.6\columnwidth]{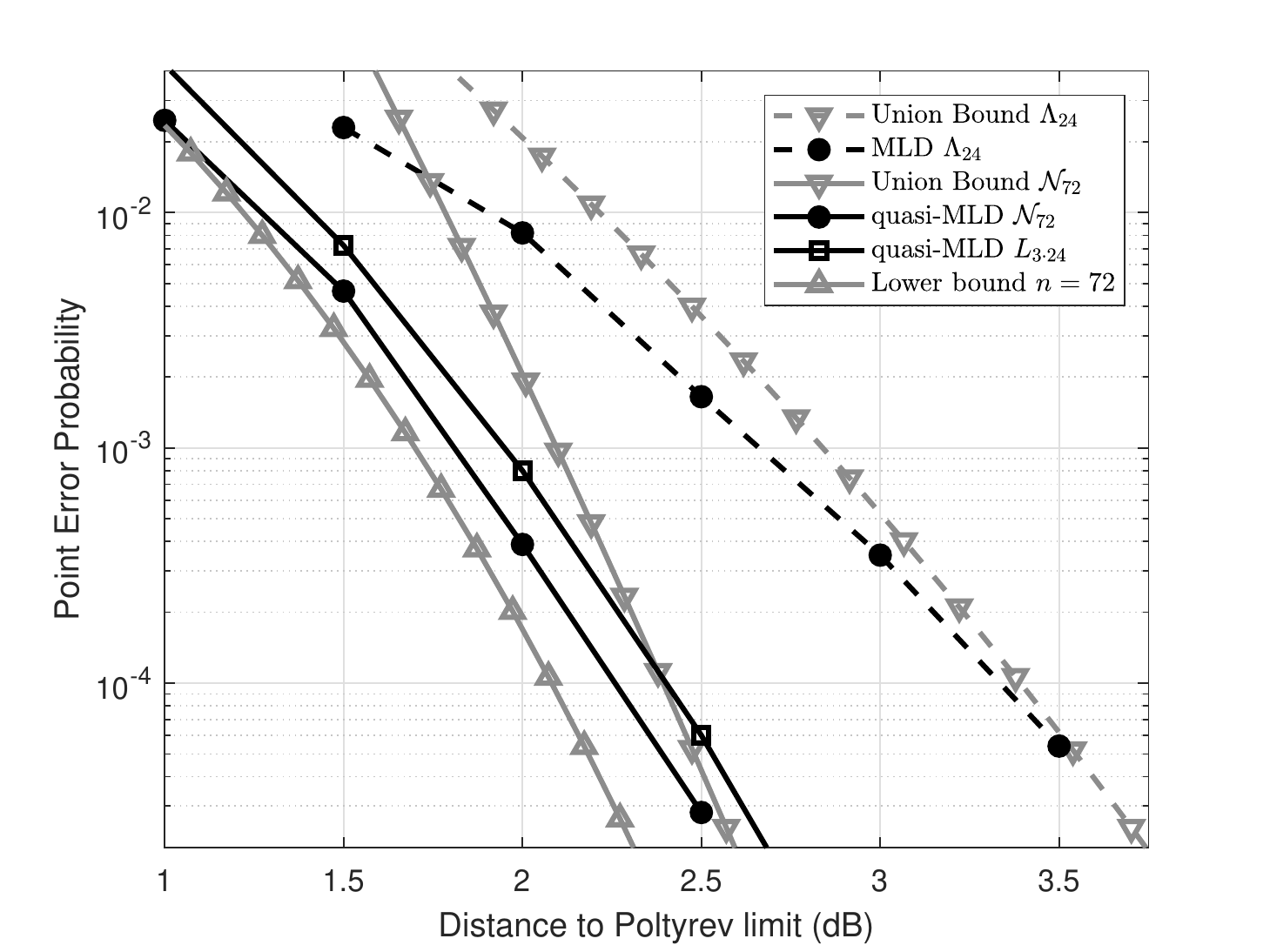}
	\caption{Performance of $\mathscr{N}_{72}$ and $L_{3 \cdot24}$ on the Gaussian channel. The union bound is computed from the two first lattice shells of $\mathscr{N}_{72}$.
The curves for $\Lambda_{24}$ are also provided for comparison.
}
\label{fig_perfNebe_bis}
\end{figure}
\section{Lattice decoding benchmark}
\label{sec_simu_gauss}

\begin{figure}[t]
\centering
\includegraphics[width=0.6\columnwidth]{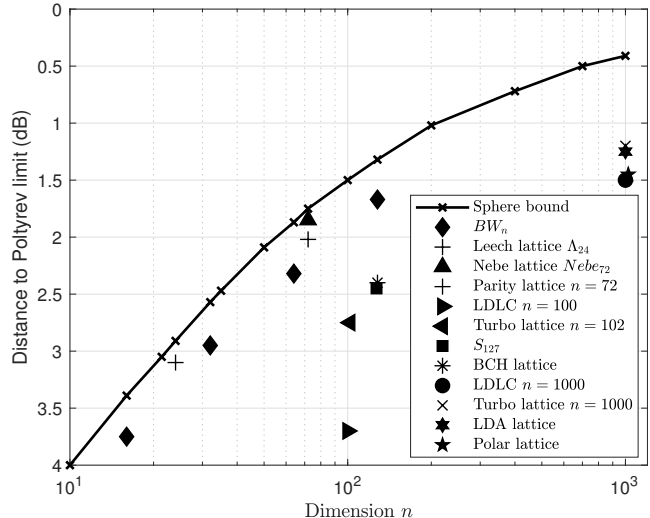}
\vspace{-2mm}
\caption{Performance of different lattices for normalized error probability $P_e=10^{-5}$.}
\label{fig_benchmark}
\end{figure}

\indent We compare the performance of lattices and decoders shown in the previous sections to existing schemes in the literature at $P_e=10^{-5}$.
For fair comparison at different dimensions, we let $P_e$ be either the symbol-error probability or the normalized error probability, which is equal to the point-error probability divided by the dimension (as done in e.g. \cite{Tarokh1999}).\\
First, several constructions have been proposed for block lengths around $n=100$ in the literature.
In \cite{Matsumine2018} a two-level construction based on BCH codes with $n=128$ achieves
this error probability at 2.4 dB. The decoding involves an OSD of order 4 with 1505883 candidates.
In \cite{Agrawal2000} the multilevel (non-lattice packing) $\mathcal{S}_{127}$ ($n=127$) has similar
performance but with much lower decoding complexity via generalized minimum distance decoding. 
In \cite{Sakzad2010} a turbo lattice with $n=102$ and in \cite{Sommer2008} a LDLC with $n=100$
achieve the error probability with iterative methods at respectively 2.75 dB, and 3.7 dB (unsurprisingly, these two schemes are efficient for larger block-lengths).
All these schemes are outperformed by $BW_{64}$, the 3-parity-Leech lattices, and $\mathscr{N}_{72}$, where $P_e=10^{-5}$ is reached at respectively 2.3 dB, 2.02 dB and 1.85 dB.
Moreover, $BW_{128}$ has $P_e=10^{-5}$ at $1.7$ dB, which is similar to many schemes with block length $n=1000$
such as the LDLC (1.5 dB) \cite{Sommer2008},  the turbo lattice (1.2 dB) \cite{Sakzad2010}, 
the polar lattice with $n=1024$ (1.45 dB) \cite{Yan2013}, and the LDA lattice (1.27 dB) \cite{diPietro2012}. 
This benchmark is summarized on Figure~\ref{fig_benchmark}.

\section{Conclusions}
\label{sec_conclu}

In this paper, we present a unified framework for building lattices. It relies on a simple parity check, which can be applied recursively and combined to repetition coding.
Famous lattices such as the Leech lattice in 24 dimensions, Nebe's extremal lattice in 72 dimensions, and Barnes-Wall lattices are obtained in this framework. 
A new decoding paradigm is established from this construction by taking into account the coset parity constraint. 
The paradigm leads to new bounded-distance decoders, list decoders, and quasi-optimal decoders on the Gaussian channel in terms of probability of error per lattice point. 
Quasi-optimal performance for $BW_{64}$, $\mathscr{N}_{72}$, and $BW_{128}$ are shown to be achievable at reasonable complexity. 
A new parity lattice $L_{3 \cdot 24}$ is also considered.
It offers an excellent performance-complexity trade-off.
The elegant single parity-check construction and its associated decoders are promising for the study of lattices in moderate and large dimensions.

\section{Appendix}
\label{sec_app_main}

\subsection{Proof of Theorem~\ref{th_turyn}}
\label{App_proof_Leech_turing}

The following proof is not new, but it enables to make a clear link between the $k$-ing construction and $\Lambda_{24}$ using our notations.
\begin{proof}
We let $E_8$ be scaled such that $d(E_8)=2$
and $\vol(E_8)=1$. This version of the Gosset lattice is even. Then, $S=\frac{1}{\sqrt{2}}E_8$ has $d(S)=1$, $\vol(S)=2^{-4}$
and $\vol(T)=\vol(T_{2\theta})=1$, $d(T)=d(T_{2\theta})=2$.
Also, $|\alpha|=|\beta|=2^4$ from~\eqref{equ_vol}.\\
Let $x=(a,b,c) \in \turing$.
Firstly, using Theorem~\ref{theo_min_dist}, we have $d(\turing)\geq3$. 
Then, assume that $a=m+t_1$ and $b=m + t_2$ (with the notations of~\eqref{equ_ning_2}) have both odd squared norms. 
This is equivalent to having the scalar products $\langle m, t_1 \rangle = \frac{\nu}{2}$ and
$\langle m, t_2 \rangle = \frac{\nu'}{2}$, where $\nu$ and $\nu'$ are integers.
Therefore, $\langle m, t_1+t_2 \rangle$ is integer and $c=m+t_1+t_2$ has an even squared norm. 
We just proved that $\turing$ is even.
This implies that $d(\turing)=4$.\\
The last step aims at proving that $\turing$ has a unit volume.
$\turing$ is obtained as the union of $|\alpha| |\beta|^2=2^{12}$ cosets of $(2S)^3$.
Hence, $\vol(\turing)=\vol((2S)^3) / 2^{12}=1$.  \\
Finally, $\Lambda_{24}$ is the unique lattice in dimension 24 with fundamental coding gain equal to 4. 
\end{proof}

\subsection{Proof of Theorem~\ref{theo_error_proba}}
\label{app_proof_error_proba}

If Step~14 is removed at the last recursive iteration of Algorithm~\ref{main_alg_pari_split_recu_noSp} the sent point $x=(x_1,x_2, ..., x_k)$ is not in the outputted list if 
\begin{itemize}
\item  $x_i \notin \mathcal{T}_i$ for at least two lists $\mathcal{T}_i$ (at Step~4 of Algorithm~\ref{main_alg_pari_split_recu_noSp}),
\item or if $x_1,..., x_{j\neq i },..., x_k \in \mathcal{T}_1,..., \mathcal{T}_j,...,\mathcal{T}_k$, and $x_i - (-\sum_{j \neq i} x_j)\notin \mathcal{V}_i$ (for at least one $i$).
\end{itemize}
Let the noise $w=(w_1,...,w_i,...,w_k)$. Due to the i.i.d property of the noise, we have $P(||w_1||^2> \frac{r}{2})=P(||w_i||^2> \frac{r}{2})$ for all $1\leq i \leq k$.
As a result, 
$P(x \notin \mathcal{T}) $ becomes 
\small
\begin{align}
\label{eq_perf_pred}
\begin{split}
P(x \notin \mathcal{T})  \leq &\binom{k}{2} P(||w_i||^2> \frac{r}{2})^2 + 	k P(||w_i||^2> r) P(||w_i||^2 < \frac{r}{2})^{k-1}, \\ 
=&\binom{k}{2} F(\frac{n}{2},\frac{r}{2},\sigma^2)^2 + 
k F(\frac{n}{2},r,\sigma^2) F(\frac{n}{2},\frac{r}{2},\sigma^2) ^{k-1}. 
\end{split}
\end{align}
\normalsize

More generally, we have
\small
\begin{align}
\label{equ_list_pb}
\begin{split}
 P(x \notin \mathcal{T})  \leq & \binom{k}{2}  P( x_j \notin \mathcal{T}_j)^2 + k P( x_i-(-\sum_{j \neq i} x_j) \notin \mathcal{V}_i) (1-P( x_j \notin \mathcal{T}_j))^{k-1}.
 \end{split}
 \end{align}
\normalsize
This idea can be recursively applied if we remove Step~14 at each recursion. 
Let $U_n(r,\sigma^2)$ denote an upper-bound of $P(x \notin \mathcal{T})$.
We have a recursion of the form

\small
\begin{align}
\label{equ_perf_noRem}
U_n(r,\sigma^2) = \binom{k}{2}  U_{\frac{n}{k}}(\frac{r}{2},\sigma^2)^2 + k \cdot U_{\frac{n}{k}}(r,\sigma^2)(1- U_{\frac{n}{k}}(\frac{r}{2},\sigma^2))^{k-1},
\end{align}
\normalsize
where we set $U_n(r,\sigma^2) =1$ if the right-hand term is greater than 1.

%

Note that $\vol(L_\frac{n}{k})^{\frac{2}{n/k}} = \vol(L_{n})^{\frac{2}{n}} / 2^{\frac{1}{k}}$, indeed:
\small
\begin{align}
\begin{split}
\label{equ_volume_pari_lat}
\vol( L_{kn})^{\frac{2}{kn}}&=\left( \frac{\vol(\theta  L_{n})^k}{|\beta|^{k-1}} \right)^{\frac{2}{kn}} 
= \left( \frac{(\vol( L_{n}) \cdot 2^{\frac{n}{2}})^k }{2^{(\frac{n}{2})^{(k-1)}}} \right)^{\frac{2}{kn}},\\
&=\vol( L_{n})^{\frac{2}{n}} \cdot 2^{\frac{1}{k}}. 
\end{split}
\end{align}
\normalsize
Moreover, we also have $d(\Gamma(V,\beta,k)_{\mathcal{P}})=d(V)=2d(T)$. 
Hence, if we express the recursion as a function of the VNR $\Delta = \frac{\vol(L_\frac{n}{k})^{\frac{2}{n/k}}}{2 \pi e \sigma^2}$ and the relative radius $\delta$, we get:
\small 
\begin{align}
\begin{split}
U_n(\delta,\Delta) = &\binom{k}{2} U_{\frac{n}{k}}(\delta,\frac{\Delta}{2^{\frac{1}{k}}})^2 +  k  U_{\frac{n}{k}}(\delta, 2^{\frac{k-1}{k}}\Delta)(1- U_{\frac{n}{k}}(\delta , \frac{\Delta}{2^{\frac{1}{k}}}))^{k-1}.
 \end{split}
\end{align}
\normalsize

With the first splitting strategy (but not the second splitting strategy) the error probability is bounded from above as

\small
\begin{align}
\begin{split}
 P(n,\sigma^2,x \notin \mathcal{T})  \leq  \binom{k}{2}  P( x_i \notin \mathcal{T}^{\delta})^2 
 &+ k  \Big[ P( x_{j } \notin \mathcal{T}^{\frac{2}{3}\delta}, \ x_{j} \in \mathcal{T}^{\delta} )^{k-1} P(x_i-(-\sum_{j \neq i} x_j) \notin \mathcal{V}^{2/3\delta}_i) \\
&+ P( x_i-(-\sum_{j \neq i} x_j) \notin \mathcal{V}^{\delta}_i)) P( x_{j } \in \mathcal{T}^{\frac{2}{3}\delta})^{k-1}\Big], 
 \end{split}
 \end{align}
\normalsize
where 
\small
\begin{align}
\begin{split}
P( x_{j } \notin \mathcal{T}^{\frac{2}{3}\delta}, \ x_{j} \in \mathcal{T}^{\delta} ) = & (1- \frac{ P(x_{j } \in \mathcal{T}^{\frac{2}{3}\delta} )}{ P(x_{j } \in \mathcal{T}^{\delta} )}) P(x_{j } \in \mathcal{T}^{\delta}) = P(x_{j } \in \mathcal{T}^{\delta}) - P(x_{j } \in \mathcal{T}^{\frac{2}{3}\delta} ).
 \end{split}
 \end{align}
\normalsize
\subsection{Proof of Lemma~\ref{lem_kiss_numb}}
\label{app_kissing}
\begin{proof}
The proof is similar to that of Theorem~3.3 in \cite{Nebe2012}.
The vectors of squared norm 8 in $\Gamma(V,T,3)_{\mathcal{P}}$ have only the following possible forms.
\begin{enumerate}
\item $(a,0,0), a \in V$ and $||a||^2 = 8$. The number of such vectors (counting the combinations) is $196560 \cdot 3$ vectors, i.e. the minimal vectors in $V^3$.
\item $(n_1,n_2,0)$, $n_1,n_2 \in T, n_1+n_2 \in V$ and $||n_1||^2=||n_2||^2=4$. The number of such vectors (counting the combinations) is $196560 \cdot 48 \cdot 3$. 
There are $196560$ possibilities for $n_1$. Given $n_1$ how many choices are they for $n_2$?
This is equivalent to asking the number of squared norm $8$ vectors in the coset $m+V$, which are therefore
congruent mod $V$. 
It is well-known (see Theorem~2 in \cite[Chap.12]{Conway1999}) that this number is 48, 24 mutually orthogonal pairs of vectors
(one can check that $|T/V|=2^{12}\cdot 48 = 196560$, the number of minimal vectors of $\Lambda_{24}$).
Hence, there are 48 choices for $n_2$. Finally, the factor 3 comes from the combinations. 
\end{enumerate}
\end{proof}

\subsection{Proof of Theorem~\ref{theo_main_complex_sp}} 
\label{app_proof_list}

To lighten the notations, we write $l(\delta)$ for $l(n/2,\delta)$ and $\mathfrak{C}(\delta)$ for $\mathfrak{C}(n/2,\delta)$. 

\noindent $\bullet$ If $\delta < \frac{3}{8}$, the decoder of \cite{Grigorescu2017}, whose complexity is given by \eqref{complex_algo_grigo}, yields $O(n^2)$.

%


\noindent $\bullet$ If $\frac{3}{8} \leq \delta <\frac{1}{2}$, the decoder of \cite{Grigorescu2017}, whose complexity is given by \eqref{complex_algo_grigo}, yields $O(n^2)$.

\noindent $\bullet$ If $\delta = \frac{3}{4} - \epsilon $, $0<\epsilon \leq 1/4$: 
Then, $\frac{2}{3} \delta < \frac{1}{2}$, $ \frac{\delta}{2} < \frac{3}{8}$, $\frac{\delta}{3} < \frac{1}{4} $. 
We have $l(\frac{2}{3}\delta) = l(\frac{\delta}{2}) = O(1)$, $l( \frac{\delta}{3})=1$. 
We get
\begin{align*}
\small
\begin{split}
\mathfrak{C}(n,\delta) = &[2l(\frac{2}{3} \delta)+2] \mathfrak{C}(\delta) + l(\delta)O(n^{2}) = l(\delta)O(n^{2}) \cdot   \sum_{i=0}^{\log_2 n} \left( \frac{2 l(\frac{2}{3} \delta)+2}{4} \right)^i, \\
=&l(\delta) O(n^{2 + \log_2 [\frac{l(\frac{2}{3} \delta)+1}{2}]} )= l(\delta) O(n^{1+\log_2 [\lfloor \frac{3}{4 \epsilon} \rfloor +1]} ).
\end{split}
\end{align*}

Consequently, if $\delta=\frac{1}{2}$, $l(\delta) \leq 2n$ and $\epsilon = \frac{1}{4}$. Then 
$\mathfrak{C}(n,\delta) = O(n^{4})$. If $\delta>\frac{1}{2}$, we have $l(\delta) = O(n^{\log_2 4 \lfloor \frac{3}{4\epsilon} \rfloor})$. Then,
$\mathfrak{C}(n,\delta) = O(n^{1+\log_2 4 \lfloor \frac{3}{4 \epsilon} \rfloor^2  })$,
where we assumed that $\frac{3}{4 \epsilon}>>1$.

\noindent $\bullet$ If $ \delta = \frac{3}{4}$:  See Appendix F in \cite{Corlay2020}  (long version on arXiv).

\noindent $\bullet$ If $\delta  = 1 - \epsilon$, $0<\epsilon<\frac{1}{4}$:  See Appendix F in \cite{Corlay2020}  (long version on arXiv).


\subsection{Proof of Theorem~\ref{theo_ave}}
\label{app_theo_ave}
The result on the complexity is obtained by adapting \eqref{equ_exple1}, \eqref{equ_exple2}, and the complexity formulas
in Theorem~\ref{theo_main_complex_sp}. 

For $3/8 <\delta \leq 9/16$, we use the fact $E_y[l'(\frac{n}{2},\frac{3}{8},y)] \geq E_y[l'(\frac{n}{4},\frac{3}{8},y)] \geq ...$.
\small
\begin{align}
\label{equ_base_recu_mod}
\begin{split}
E_y[l'(n,\delta,y)] \leq  2 \big[E_y[l'(\frac{3}{8},y)]  l(\delta ) + l(\delta)  E_y[l(\frac{3}{8},y)] ] \leq & 4  E_y[l'(\frac{3}{8},y)] l(\delta ) = O(n^{\log_2(4 E_y[l'(\frac{3}{8},y)])}), \\
 = &O(n^{2 + \log_2 E_y[l'(\frac{3}{8},y)]}). 
\end{split}
\end{align}
\normalsize

\subsection{A proof that $R(\mathscr{N}_{72})>\sqrt{2} \rho(\mathscr{N}_{72})$}
\label{app_cov}

\begin{lemma}
$R(\mathscr{N}_{72})>\sqrt{2} \rho(\mathscr{N}_{72})$.
\end{lemma}
The proof of this lemma is due to Gabriele Nebe (private communication).
\begin{proof}
Let $\mathscr{N}_{72}$ be scaled such that $\rho(\mathscr{N}_{72}) = \sqrt{2}$.  
The proof is done by contradiction.
Assume 
that $R(\mathscr{N}_{72})~=~\sqrt{2}\rho(\mathscr{N}_{72})=2$.
Then, for any point $1/2 v \in 1/2 \mathscr{N}_{72}$,
there is a point $x \in \mathscr{N}_{72}$ with $|| x-1/2v ||  \le 2$.
Squaring leads to $|| 2x - v || ^2  \le 16 $.
So each of the $2^{72}$ cosets of $2\mathscr{N}_{72}$ in $\mathscr{N}_{72}$ has to contain
a point $w= 2x-v$ of squared norm  smaller or equal to 16. 

Now $\mathscr{N}_{72}$ has exactly 107502190683149087281 pairs ${ \pm w }$ of squared norm $\le 16$ (obtained from the theta series of $\mathscr{N}_{72}$).
This number is smaller than $|\mathscr{N}_{72}/2\mathscr{N}_{72}|$.
Hence the covering radius of $\mathscr{N}_{72}$ is strictly larger than~2.
\end{proof}

\newpage


\end{document}